\pdfoutput=1
\documentclass[a4paper,10pt]{article}
\usepackage{amsfonts,amsmath,amssymb,amstext,adjustbox}
\usepackage{mathrsfs}

\usepackage[matha,mathb]{mathabx}
\usepackage{mathtools}

\usepackage{ifthen} %
\usepackage{xifthen} %

\makeatletter
\newcommand*{\da@rightarrow}{\mathchar"0\hexnumber@\symAMSa 4B }
\newcommand*{\da@leftarrow}{\mathchar"0\hexnumber@\symAMSa 4C }
\newcommand*{\xdashrightarrow}[2][]{%
  \mathrel{%
    \mathpalette{\da@xarrow{#1}{#2}{}\da@rightarrow{\,}{}}{}%
  }%
}
\newcommand{\xdashleftarrow}[2][]{%
  \mathrel{%
    \mathpalette{\da@xarrow{#1}{#2}\da@leftarrow{}{}{\,}}{}%
  }%
}
\newcommand*{\da@xarrow}[7]{%
  \sbox0{$\ifx#7\scriptstyle\scriptscriptstyle\else\scriptstyle\fi#5#1#6\m@th$}%
  \sbox2{$\ifx#7\scriptstyle\scriptscriptstyle\else\scriptstyle\fi#5#2#6\m@th$}%
  \sbox4{$#7\dabar@\m@th$}%
  \dimen@=\wd0 %
  \ifdim\wd2 >\dimen@
    \dimen@=\wd2 %
  \fi
  \count@=2 %
  \def\da@bars{\dabar@\dabar@}%
  \@whiledim\count@\wd4<\dimen@\do{%
    \advance\count@\@ne
    \expandafter\def\expandafter\da@bars\expandafter{%
      \da@bars
      \dabar@
    }%
  }%
  \mathrel{#3}%
  \mathrel{%
    \mathop{\da@bars}\limits
    \ifx\\#1\\%
    \else
      _{\copy0}%
    \fi
    \ifx\\#2\\%
    \else
      ^{\copy2}%
    \fi
  }%
  \mathrel{#4}%
}
\makeatother

\newcommand{\parleftarrow}[2][]{%
  \ifthenelse{\isempty{#1}}%
  {\ensuremath{\xdashleftarrow{\ensuremath{#2}}}}%
  {\ensuremath{\xdashleftarrow{\parbox{#1}{\centering\ensuremath{#2}}}}}%
}
\newcommand{\parrightarrow}[2][]{%
  \ifthenelse{\isempty{#1}}%
  {\ensuremath{\xdashrightarrow{\ensuremath{#2}}}}%
  {\ensuremath{\xdashrightarrow{\parbox{#1}{\centering\ensuremath{#2}}}}}%
}

\usepackage{fancybox}
\newcommand{\plrightarrow}[2][]{%
  \ifthenelse{\isempty{#1}}%
  {\ensuremath{\xrightarrow{\ensuremath{#2}}}}%
  {\ensuremath{\xrightarrow{\parbox{#1}{\centering\ensuremath{#2}}}}}%
}
\newcommand{\plleftarrow}[2][]{%
  \ifthenelse{\isempty{#1}}%
  {\ensuremath{\xlefttarrow{\ensuremath{#2}}}}%
  {\ensuremath{\xleftarrow{\parbox{#1}{\centering\ensuremath{#2}}}}}%
}
\newlength{\commlength}
\newcommand{\pleftarrow}[1]{\parleftarrow[\commlength]{#1}}
\newcommand{\prightarrow}[1]{\parrightarrow[\commlength]{#1}}

\makeatletter
\newcommand{\incircbin}{\mathpalette\@incircbin}
\newcommand{\@incircbin}[2]{\mathbin{\ooalign{\hidewidth$#1#2$\hidewidth\crcr$#1\ovoid$}}}
\newcommand{\oeq}{\incircbin{=}}
\makeatother

\usepackage{amsthm}
\newtheorem{theorem}{Theorem}
\newtheorem{lemma}[theorem]{Lemma}
\newtheorem{remark}[theorem]{Remark}

\newtheorem{definition}[theorem]{Definition}

\usepackage{fullpage}
\usepackage{thm-restate,apptools}

\newenvironment{restatablebackref}[4]{\restatable[\IfAppendix{\normalfont{}From
    page~\pageref{#3}}{\normalfont{}A proof is given in \cref{#4}}]{#1}{#2}\label{#3}}{\endrestatable}

\usepackage{amsfonts,adjustbox}
\usepackage{enumitem}

\usepackage{tabularx,multirow,booktabs}
\usepackage{arydshln} %

\newcommand{\Z}{\ensuremath{\mathbb{Z}}\xspace}
\newcommand{\N}{\ensuremath{\mathbb{N}}\xspace}

\newcommand{\GG}{\ensuremath{\mathbb{G}}\xspace}
\newcommand{\tmod}{\ensuremath{\,\mathrm{mod}\,}}
\newcommand{\bigSomething}[2]{\ensuremath{#1\mathopen{}\left(#2\right)\mathclose{}}\xspace{}}
\newcommand{\bigO}[1]{\bigSomething{\mathcal{O}}{#1}}
\newcommand{\bigOsqrt}[1]{\ensuremath{\mathcal{O}(\sqrt{#1})\xspace{}}}
\newcommand{\smallo}[1]{\bigSomething{o}{#1}}

\newcommand{\random}{\stackrel{\$}{\leftarrow}}
\newcommand{\checks}[1]{\ensuremath{\mathrel{\stackrel{?}{#1}}}\xspace}
\newcommand{\Primes}{\ensuremath{{\mathbb P}}\xspace}

\newcommand{\R}{\ensuremath{{\mathcal R}}\xspace}

\newcommand{\pk}{\ensuremath{pk}\xspace}
\newcommand{\sk}{\ensuremath{sk}\xspace}

\usepackage{pifont}
\usepackage{savesym}
\savesymbol{checkmark}
\newcommand{\gcheck}{{\color{darkgreen}\checkmark}}
\usepackage{dingbat}
\newcommand{\xyes}{\gcheck}
\newcommand{\xno}{\color{darkred}\ding{55}}

\newcommand{\vect}[1]{\ensuremath{{#1}}}
\newcommand{\matr}[1]{\ensuremath{{#1}}}
\newcommand{\uu}{\vect{u}}
\def\vv{\vect{v}} %
\newcommand{\xx}{\vect{x}}
\newcommand{\yy}{\vect{y}}

\newcommand{\WW}{\matr{W}}
\newcommand{\svec}{\vect{\gamma}}
\newcommand{\MM}{\matr{M}}

\newcommand{\AAA}{\matr{A}}
\newcommand{\BB}{\matr{B}}
\newcommand{\ww}{\vect{w}}
\newcommand{\client}{\ensuremath{\mathcal{C}}}
\newcommand{\verifier}{\ensuremath{\mathcal{V}}}
\newcommand{\server}{\ensuremath{\mathcal{S}}}
\newcommand{\clstate}{\ensuremath{st_\client}\xspace}
\newcommand{\servstate}{\ensuremath{st_\server}\xspace}
\newcommand{\verstate}{\ensuremath{st_\verifier}\xspace}
\newcommand{\ttb}[1]{\ensuremath{\text{\bf\texttt{#1}}}}

\newcommand{\negl}{\mathsf{negl}}
\newcommand{\Proba}[1]{\ensuremath{{\mathcal{P}r}\left[#1\right]}}
\newcommand{\Init}{\ttb{Init}\xspace}
\newcommand{\Setup}{\ttb{Setup}\xspace}
\newcommand{\SetupSet}{\ensuremath{\mathscr{S}}}
\newcommand{\deltaSetup}{\ttb{\ensuremath{\delta}Setup}\xspace}
\newcommand{\Read}{\ttb{Read}\xspace}
\newcommand{\Verif}{\ttb{Verify}\xspace}
\newcommand{\Update}{\ttb{Update}\xspace}
\newcommand{\Write}{\ttb{Write}\xspace}
\newcommand{\deltaUpdate}{\ttb{\ensuremath{\delta}Update}\xspace}
\newcommand{\Eval}{\ttb{Eval}\xspace}
\newcommand{\Audit}{\ttb{Audit}\xspace}
\newcommand{\accept}{\ttb{accept}\xspace}
\newcommand{\reject}{\ttb{reject}\xspace}
\newcommand{\adversary}{\ensuremath{\mathcal{A}}}
\newcommand{\observer}{\ensuremath{\mathcal{O}}}

\newcommand{\gensym}{\ensuremath{e(g;g)}}
\newcommand{\gen}{\ensuremath{e(g_1;g_2)}}
\newcommand{\convertintopoly}{Vect{\bf 2}Poly}

\newenvironment{smatrix}{\left(\begin{smallmatrix}}{\end{smallmatrix}\right)}

\newcommand{\Transpose}[1]{{{\matr{#1}}^{\intercal}}\xspace}

\newcommand{\narrowfont}[1]{\scalebox{.75}[1.0]{\textbf{\footnotesize{#1}}}}

\newcommand{\mtrootfrompath}{\narrowfont{MTpathRoot}\xspace}
\newcommand{\mtrootfromleaves}{\narrowfont{MTRoot}\xspace}
\newcommand{\mtcreate}{\narrowfont{MTTree}\xspace}
\newcommand{\mtuncles}{\narrowfont{MTUncles}\xspace}
\newcommand{\mtupdateleaf}{\narrowfont{MTupdLeaf}\xspace}

\newcommand{\compsec}{\ensuremath{\kappa}}

\date{\today{}}
\title{{VESPo}:
  Verified Evaluation of Secret Polynomials
  (with~application~to~dynamic~proofs~of~retrievability)}

\newcommand{\email}[1]{\href{mailto:#1}{\nolinkurl{#1}}}
\newcommand{\emails}[2]{\href{mailto:#2}{\nolinkurl{#1}}}

\author{{Jean-Guillaume Dumas}%
\footnote{%
  {Universit\'e Grenoble Alpes},
  {Laboratoire Jean Kuntzmann, UMR CNRS 5224, Grenoble INP}.
  {700 avenue centrale, IMAG --- CS 40700},
  {38058 Grenoble},
  {France}.
\emails{{Jean-Guillaume.Dumas,Aude.Maignan,Clement.Pernet}@univ-grenoble-alpes.fr}{Jean-Guillaume.Dumas@univ-grenoble-alpes.fr,Aude.Maignan@univ-grenoble-alpes.fr,Clement.Pernet@univ-grenoble-alpes.fr}.
}
\and{Aude Maignan}\footnotemark[2]
\and{Cl\'ement Pernet}\footnotemark[2]
\and{Daniel S.\ Roche}
\footnote{%
	{United States Naval Academy},
	{Annapolis},
	{Maryland},
	{United States}.
\email{Roche@usna.edu}.
}
}

\newcommand{\myqed}{}

\usepackage{xcolor,svgcolor}
\usepackage{hyperref}
\makeatletter
\hypersetup{
pdftitle={{VESPo}: Verified Evaluation of Secret Polynomials},
pdfauthor={Jean-Guillaume Dumas, Aude Maignan, Cl\'ement Pernet, Daniel S.\ Roche},
breaklinks=true,
plainpages=true,
colorlinks=true,
 linkcolor=darkblue,
 citecolor=darkgreen,
 urlcolor=darkred,
}
\makeatother
\usepackage[all]{hypcap} %

\usepackage[capitalise,noabbrev,nameinlink]{cleveref}

\usepackage{algorithm}
\usepackage{algcompatible}
\newcommand{\TO}{\textbf{to}\xspace}
\newcommand{\algorithmicreturn}{\textbf{return}}
\newcommand{\RETURN}{\STATE\algorithmicreturn{}\xspace}
\newcommand{\CRETURN}[1]{\STATE #1: \algorithmicreturn{}\xspace}

\usepackage{xspace}
\algnewcommand{\IfThen}[2]{%
  \State \algorithmicif\ #1\ \algorithmicthen\ #2}
\algnewcommand{\IfThenElse}[3]{%
  \State \algorithmicif\ #1\ \algorithmicthen\ #2\ \algorithmicelse\ #3}
\algblock{ParFor}{EndParFor}
\algnewcommand\algorithmicparfor{\textbf{parfor}}
\algnewcommand\algorithmicpardo{\textbf{do}}
\algnewcommand\algorithmicendparfor{\textbf{end\ parfor}}
\algrenewtext{ParFor}[1]{\algorithmicparfor\ #1\ \algorithmicpardo}
\algrenewtext{EndParFor}{\algorithmicendparfor}
\makeatletter
\newcounter{algorithmicH}%
\let\oldalgorithmic\algorithmic
\renewcommand{\algorithmic}{%
  \stepcounter{algorithmicH}%
  \oldalgorithmic}%
\renewcommand{\theHALG@line}{ALG@line.\thealgorithmicH.\arabic{ALG@line}}
\makeatother

\begin{document}
\maketitle

\begin{abstract}
Proofs of Retrievability are protocols which allow a Client to store data remotely and to efficiently ensure, via audits, that the entirety of that data is still intact. Dynamic Proofs of Retrievability (DPoR) also support efficient retrieval and update of any small portion of the data.
We propose a novel protocol for arbitrary outsourced data storage that achieves both low remote storage size and audit complexity.
A key ingredient, that can be also of intrinsic interest, reduces to efficiently evaluating a secret polynomial at given public points, when the (encrypted) polynomial is stored on an untrusted Server.
The Server performs the evaluations and also returns associated certificates. A Client can check that the evaluations are correct using the certificates and some pre-computed keys, more efficiently than re-evaluating the polynomial.
Our protocols support two important features: the polynomial itself can be encrypted on the Server, and it can be dynamically updated by changing individual coefficients cheaply without redoing the entire setup.
Our methods rely on linearly homomorphic encryption and pairings, and our implementation shows good performance for polynomial evaluations with millions of coefficients, and efficient DPoR with terabytes of data.
For instance, for a 1TB database, compared to the state of art, we can reduce the Client storage by 5000x, communication size by 20x, and client-side audit time by 2x, at the cost of one order of magnitude increase in server-side audit time.
\end{abstract}

\section{Introduction}
With a constant growth in the amount of produced data, it becomes more
and more important to use remote facilities to store this data.
Users and organizations using such outsourcing need to ensure the
\emph{integrity} of their data.

In this setting, a Client wishes to store her data on an untrusted
Server, then verify (without full retrieval) that the Server still
stores the data intact. The crucial protocol is an $\Audit$, wherein the
Client issues some challenge to the Server, then verifies the response
using some pre-computed information to prove that the original data is
still recoverable in its entirety.
This is the field of \textbf{Proofs of Retrievability}~(PoR),
somewhat overlapping with the problem of
\emph{Provable Data
  Possession}~(PDP)~\cite{juels2007pors,ateniese2007provable}.

A variety of tools have been employed to develop efficient PoR and
PDP protocols, see for
instance~\cite{juels2007pors,ateniese2007provable,SachamPOR08,CashPOR13,Shi:2013:orampor,Anthoine:hal-02875379}
and references therein.
Retrievability is proven when
any sequence of successful audits can, with high probability, be used
to recover the original data, e.g., by polynomial interpolation; thus
any Server with a good chance to pass a single random audit must hold the
entire data intact.
Note that this recovery mechanism is not actually crucial except to
\emph{prove} the soundness of the audit protocol; the important feature
is how cheaply the audits can be performed by a Server and
resource-constrained Client.

Some of these protocols are based on verifiable computing, so that a
PoR audit consists of some verified computation over the stored
data~\cite{Gennaro:2010:outsourcing}. Generally speaking, verifiable
computing consists in
delegating the computation of a function to an untrusted Server.
This Server returns the result as well as a proof of its correctness,
and verifying a result should be less expensive than computing it
directly.
While certified and verified
computation protocols date back decades, the practical need for
efficient methods is especially evident in cloud computing,
wherein again a low-powered device,
such as a mobile phone, may wish to outsource expensive and critical
computations to an untrusted, shared-resource, commercial cloud. %
The literature on verifiable computation protocols can be
divided into general-purpose computations --- of an
arbitrary algebraic circuit --- and more limited but
more efficient special-purpose computations of certain
functions (see, e.g.,~\cite{Walfish:2015:VCWRT,jgd:2018:outsourcing}
and references therein).
In the latter category, one problem is Verifiable
Polynomial Evaluation (VPE), where a Client wishes to outsource the
evaluation of a univariate polynomial $P$ on an untrusted Server at
given public points and efficiently verify the result.

Existing VPE protocols usually do not consider dynamicity, at least
not efficiently: even for the modification of a single coefficient,
most of the time the whole protocol has to be reinitialized.
Also, previous PoR protocols would either have a low audit complexity
but a storage size several times that of the database; or have
low remote storage but a less scalable audit complexity.
In this paper, we propose a novel protocol for arbitrary outsourced
dynamic data storage that achieves both low remote storage size and audit
complexity.
A key ingredient of our protocol is to be able to perform a
dynamic VPE, in order to efficiently handle updates of the database.

A verifiable polynomial evaluation scheme is conventionally
composed of three main algorithms.
First, a Client runs $\Setup(P)$ to compute some public
representation of the (potentially secret) polynomial $P$ (which may
be stored on the Server) as well as
some private information which will be used to verify later evaluations.
This step may be somewhat expensive, but only needs to be performed
once.
The second algorithm, $\Eval(x)$, is run by the
Server using a public evaluation point $x$ provided by the Client. The
Server produces the evaluation $y = P(x)$ as well as some proof (or
certificate) $\Pi$ that this evaluation is correct.
Finally, the third algorithm, $\Verif(y,\Pi)$, is run
by the Client to check the correctness of the evaluation. This
verification should be \emph{always correct}
and \emph{probabilistically sound}, meaning that an honest Server can
always produce a result $y$ and proof $\Pi$ that will pass the
verification, whereas an incorrect evaluation $y$ will fail the
verification with high probability for any purported proof $\Pi$. Furthermore, the
\Verif{}~algorithm should be efficient, ideally much cheaper in time
and/or space than the computation itself.

In the simplest case, the considered polynomial $P$ is static and stored
in cleartext by both the Server and the Client.
But constraints can then be added to this framework, when needed:
\begin{itemize}[leftmargin=2\labelsep]
\item
{\bf Polynomial outsourcing}.
When the Client device has limited storage, or to facilitate
evaluations for multiple Clients, both the polynomial storage and its
computation must be externalized.
Besides evaluation and verification, an additional $\Read$ protocol
is often provided to allow random access to some polynomial coefficients.
The challenge is for the Client to obtain the polynomial evaluation
while minimizing the communication costs required to verify it.
\item
{\bf Secret polynomial}.
To guarantee data privacy, the polynomial could be
hidden from the Server, or the Client, or both.
Typically, the polynomial will be
stored under a fully- or partially-homomorphic encryption scheme, in
such a way that the Server can still compute the (necessarily encrypted)
evaluation and certificate for verification.
This setting has been extensively studied in
the literature, with both general-purpose protocols as well as some
specific ones for verified polynomial evaluation; see, e.g.,~\cite{Fiore:2014:CCS,Groth:2016:size,BenSasson:2018:RS-IOP,Bunz:2018:bulletproof,Maller:2019:sonic,Fiore:2020:bvced,Bhadauria:2020:ligero,Lee:2021:tccDory,Rafols:2021:uuSNARK,Bois:2021:verifencrypt}.
\item
{\bf Dynamic updates}.
The initial $\Setup$ protocol requires knowledge of the entire
polynomial and generally is much
more costly than running $\Verif$.
This creates a challenge when the Client wishes to update only a few of
the coefficients of the polynomial.
A \emph{dynamic} VPE protocol allows for such updates efficiently.
Namely, the Client and Server storing polynomial $P=\sum_{i=0}^dp_iX^i$ for verified
evaluation can engage in an additional $\Update(i,p'_i)$ protocol,
which effectively updates $P(x)$ to $P(x) + (p'_i-p_i)x^i$ for future
evaluations.
To the best of our knowledge, no prior work in the literature discusses
dynamic updates for verified polynomial evaluation.
When the polynomial (as well as any update)
needs to be hidden from the Server,
the difficulty is in general to preserve both secrecy and verifiability while
allowing those efficient partial updates.
The importance of allowing efficient updates is motivated by our
application to verifiable data storage, where a Client outsourcing
storage of a large database wishes to make small changes efficiently.
\item
{\bf Private/public verification}.
The verification protocol is said to be \emph{private} when only the party
which holds the secrets derived during $\Setup$ can verify
evaluations. That is, any potential Verifiers (sometimes called
\emph{readers}) must be trusted not to divulge secret information to the
untrusted Server.
Sometimes, it is desirable also to have untrusted Verifiers,
who can check the result of an evaluation without knowing any secrets.
In this \emph{public verification} setting, the Client at setup time
publishes some additional information, distributed reliably but
insecurely to any Verifiers, which may be used to check
evaluations and proofs issued by the Server.
\end{itemize}

\subsection{Our contributions}
Our contributions are the following:
\newcounter{myenum}
\begin{itemize}[leftmargin=2\labelsep]
\item An (unencrypted) Verifiable Polynomial
  Evaluation (VPE) scheme with public verification, supporting
  \emph{secured dynamic updates}%
  ~ (\cref{sec:dynamic,protoDynClear}).
  The polynomial is stored in cleartext on the Server, and the technique
  used to provide a correct and sound protocol uses both Merkle trees
  and pairings. A Horner-like evaluation scheme is used to optimize the
  evaluation of the difference polynomial for the proof, and no
  secrets are required to perform the verification.
\item A novel \emph{encrypted, dynamic and private} VPE protocol
    (\cref{sec:full,proto:full}). That is, the
  polynomial is stored encrypted on the Server, and efficient updates
  to individual coefficients can be performed.
  This is achieved by combining a linearly homomorphic cryptosystem with
  techniques from the first scheme.
  Note however, this
  scheme does not support public verification as this verification now
  requires some secrets from the Client.
\item A new Dynamic Proofs of Retrievability (DPoR) scheme that is the first to
  simultaneously support small Server
  storage, dynamic updates, and efficient audits
  (\cref{sec:por,protoPor}), based on
  our novel encrypted, dynamic VPE protocol.
  Previous work either had poly-logarithmic time
  audits and linear extra storage, or sub-linear extra storage and
  polynomial-time audits; ours is the
  first to achieve both sub-linear extra storage and optimal $\bigO{\log n}$
  Client time for updates and audits.
  This could be beneficial especially in blockchain settings such as
  FileCoin where the proof and verification must be done
  on-chain~\cite{ProtocolLabs:2017:Filecoin}.
\item A full implementation and experimental timings based on our
  encrypted VPE and dynamic PoR protocols that indicate VPE up to
  millions of coefficients and DPoR up to terabytes of data, both with
  Client cost less than a few milliseconds
  (\cref{tab:lintests,table:results}).
\end{itemize}

These contributions are organized as follows.
A complete security definition of verifiable polynomial evaluation
can be found in \cref{sec:secu}.
This definition follows previous results, with the novel inclusion of
an $\Update$ protocol.
Then \cref{sec:cipher} introduces the tools for verification
of polynomial evaluation.
A motivating example is presented in the form of a direct extension of
the bilinear pairing scheme of~\cite{Kate:2010:KZG},
now supporting an encrypted input polynomial (\cref{sec:cipher,proto:cKZG}).
Since the privacy of this protocol is not proven and it
supports neither public verifiability nor dynamic updates, it
motivates the more involved contributions of \cref{sec:dynamic} (for
public verifiability and dynamicity, but on an unciphered polynomial)
and of \cref{sec:full} (for dynamicity on a ciphered polynomial, but
without public verifiability).

The efficiency of our protocols is measured by the computational
complexity of the Server-side $\Eval$ algorithm, the volume of
persistent Client storage, and the amount of communication and
Client-side complexity to perform a $\Verif$.
Improving on previously-known results, our VPEs protocols all have
\bigO{d} (parallelizable) Server-side computation, \bigO{1}
communication and Client-side computation time, and \bigO{1}
Client-side persistent storage.
We include some practical timings in
\cref{ssec:full-exper,ssec:por-exper,app:parprefix}.
In addition, our new dynamic proofs of retrievability require only
\smallo{d} extra Server space. This improves on \cite{Shi:2013:orampor}
in terms of Server storage and on \cite{Anthoine:hal-02875379} in
terms of communication and Client computation complexity for
$\Audit$.
For instance on a 1TB size database, with a Server extra storage lower
than $0.08$\%, and a Client persistent storage less than one KB,
our Client can check in less than $7$ms that their entire
outsourced data is fully recoverable from the cloud Server.

\subsection{Related work}

While ours is the first work we are aware of which considers verifiable
polynomial computation while hiding the polynomial from the Server and
allowing efficient dynamic updates, there have been a number of prior
works on different settings of the VPE problem.

One line of work considers \emph{commitment schemes} for polynomial
evaluation
\cite{Catalano:2013:vectorcommit,Camenisch:2015:asiacrypt,Libert:2016:icalp,Gabizon:2019:PLONK,Tomescu:2020:aggregatable,Boneh:2020:multiplepoints,Ozdemir:2020:usenix,Fiore:2020:bvced,Lee:2021:tccDory}.
There, the polynomial $P$ is known to the Server, who publishes a
binding commitment. The Verifier then
confirms that a given evaluation is consistent with the pre-published
commitment.
By contrast, our
protocols aim to \emph{hide the polynomial $P$ from the Server}.

Another line of work considers polynomial evaluation as an encrypted
function, which can be evaluated at any chosen point. Function-hiding
inner product encryption (IPE) \cite{BishopIPE15,KimIPE18,AbdallaIPE20}
can be used to perform polynomial evaluation without revealing the
polynomial $P$, but this inherently requires linear-time for the Client,
who must compute the first $d$ powers of the desired evaluation point
$x$.
Similarly, protocols using
a Private Polynomial Evaluation (PPE) scheme have been developed in
\cite{10.1007/978-3-319-68637-0_29}. This primitive, based on an ElGamal
scheme, ensures that the polynomial is protected and that the user is
able to verify the result given by the Server. Here the aim of the
protocol is not to outsource the polynomial evaluation, but to
obtain $P(x)$ and a proof without knowing anything about the polynomial.
To check the proof, as with IPE the Client has to perform a computation which is linear in the degree of $P$.

A third and more general approach which can be applied to the VPE
problem is that of secure evaluation of arithmetic circuits.
These protocols make use of fully homomorphic encryption (FHE) to
outsource the evaluation of an arbitrary arithmetic circuit without
revealing the circuit itself to the Server.
The VC Scheme of  \cite{Gennaro:2010:outsourcing}
is based on Yao's label construction. $P$ is first transformed into an
arithmetic circuit. The circuit is garbled once in a setup phase and
sent to the Server.
To later perform a verified evaluation, the Client sends an encryption
of $x$, the Server computes $P(x)$ through the garbled circuit, and
the Client can verify the result in time proportional to the circuit
depth, which for us is \bigO{\log d}.

Using similar techniques,
Fiore et al. and Elkhiyaoui et
al.
\cite{Benabbas2011VerifiableDO,Fiore:2012:PVD,Elkhiyaoui2016EfficientTF}
propose high-degree verified polynomial evaluations.
The major issue for these works is that they were not meant to be
dynamic: they use some structured masking that must be updated
together with the polynomial update (otherwise updates leak some
secrets). But then the update is not efficient anymore as the
structure impacts all of the polynomial coefficient masking.

More recently, Fiore et
al. \cite{Fiore:2014:CCS,Fiore:2020:bvced,Bois:2021:verifencrypt}
propose a new protocol for more general circuits, using succinct
non-interactive arguments of knowledge (SNARKs) or probabilistically
checkable proofs (PCPs) over a quotient polynomial ring.
In contrast to our work, these protocols use more expensive
cryptographic primitives, and they do not consider the possibility of
efficiently updating the polynomial -- while preserving the security
properties.
A summary of how our protocols compare to the state of the art is
given in~\cref{tab:compare}.
\begin{table}[htbp]\small
\renewcommand{\arraystretch}{0.75}\setlength{\tabcolsep}{2pt}
\caption{Comparing verifiable computation schemes for polynomial
  evaluation of degree $d$. See also~\cite[Table~1]{Maller:2019:sonic} or
  \cite[Table~1]{Rafols:2021:uuSNARK} {\footnotesize (Most of the time
    dynamicity is not considered in the literature)}.}\label{tab:compare}
\begin{tabular}{lcccccc}
\toprule
Protocol & Server & Comm. & Verif. & Dyn. & LHE & P-Q \\
\midrule
BGV11~\cite{Benabbas2011VerifiableDO}& \bigO{d} & \bigO{1} & \bigO{1} & \xno & \xyes & \xno\\
FG12~\cite{Fiore:2012:PVD} & \bigO{d} & \bigO{1} & \bigO{1} & \xno & \xno & \xno\\
libsnark~\cite{Groth:2016:size} &  \bigO{d\log{d}} & \bigO{1} & \bigO{1} &
\xno & MT & \xno \\
bulletproof~\cite{Bunz:2018:bulletproof} &  \bigO{d} &
\bigO{\log{d}} &\bigO{d}  & \xno & \xyes & \xno \\
FGP14~\cite{Fiore:2014:CCS} &  \bigO{d\log{d}} &  \bigO{1} & \bigO{\log{d}}
&\xno & FHE & \xno\\
libiop~\cite{BenSasson:2018:RS-IOP}  &  \bigO{d\log{d}}& \bigO{\log^2{d}}  & \bigO{d}
&\xno & \xno & \xyes \\
ligero++~\cite{Bhadauria:2020:ligero} & \bigO{d\log{d}}& \bigO{\log^2{d}}  & \bigO{d}
&\xno & \xno & \xyes \\
FNP20~\cite{Fiore:2020:bvced} &  \bigO{d\log{d}}  &  \bigO{1} &
\bigO{\log{d}} &\xno & FHE & \xno \\
DORY~\cite{Lee:2021:tccDory}  &   \bigO{d\log{d}} &  \bigO{\log{d}} &
\bigO{\log{d}} & \xno & \xyes & \xno \\
BCFK21~\cite{Bois:2021:verifencrypt}&  \bigO{d} &  \bigO{\log^{2}{d}} &
\bigO{\log^{2}{d}}  & \xno & FHE & \xyes\\
VESPo, \cref{proto:full}  & \bigO{d} &  \bigO{1} & \bigO{1} & \xyes & \xyes & \xno\\
\bottomrule
\end{tabular}
\end{table}

From this table, we see that many instances, like SNARKS, need
$O(d\log{d})$ operations on the Server side, where VESPo remains
linear, $O(d)$ in the input size.
Also dynamicity is usually not considered in the literature.
For us, the difficulty is to be able to modify a small part of
the input, without having to replay the whole \Setup~phase, while not
compromising security.
A salient point is that many schemes cannot directly handle the
\emph{encrypted} setting. In some cases a solution could be to simulate the
whole encryption as arithmetic circuits, but this drastically affects
performance.
For instance we tried libsnark over a Paillier encryption, this rapidly
exhausted the RAM of our server (i.e. even with degrees as small as
20), and thus denote this exhaustion by MT (memory thrashing).
Finally, we mention if the protocol is feasibly post-quantum (P-Q) secure
in the 'P-Q' column (\cite{BenSasson:2018:RS-IOP,Bhadauria:2020:ligero}
are P-Q-secure, \cite{Bois:2021:verifencrypt} do not mention it but
seems P-Q-secure, all the others, including us, use bilinear pairings).
In~\cref{app:pqc}, we abstract the requirements of our protocols to see
if they could be modified to use only post-quantum secure routines. Our
preliminary results there show that this might be possible but that using
quantum-safe routines in our case would still be several orders of
magnitude slower.

In fact, efficiency, linearity, dynamicity and encryption, are all
four of paramount importance for instance for our particular
application, as detailed next.

Proof of retrievability (PoR) and Provable data possession (PDP)
protocols also have an extensive
literature~\cite{ateniese2007provable,Erway,juels2007pors,Sebe:2008:EfficientRD,SachamPOR08,stefanov2012iris,CashPOR13,Shi:2013:orampor,Cash:2017:DPR,Anthoine:hal-02875379}.
PDPs, first introduced by {Ateniese et al.}, generally optimize Server
storage and efficiency at the cost of soundness:
a PDP audit only guarantees (probabilistically) that a
\emph{large fraction} of the data was not altered; a single block
deletion or alteration is likely to go undetected in an audit.

PoRs have stronger soundness guarantees,
but at the expense of larger and more complicated Server storage,
often based on erasure codes and/or ORAM techniques.

PoR methods based on block erasure encoding are a class of methods
which guarantee with a high probability that  the client's entire data
can be retrieved.  %
The idea is to check the authenticity of a number of erasure encoding
blocks during the data recovery step but also during the audit
algorithm. Those approaches will not detect a small amount of corrupted
data. But the idea is that if there are very few corrupted blocks,
they could be easily recovered via the error correcting
code~\cite{Lavauzelle:2016:ldcpor}.
Now, state-of-the-art, dynamic, PoR protocols either incur a
constant-factor blowup in Server storage with poly-logarithmic audit
cost \cite{CashPOR13,Cash:2017:DPR,Shi:2013:orampor},
or use negligible extra Server
storage space but require polynomial-time audits on the Client and
Server \cite{SachamPOR08,Anthoine:hal-02875379}.
We refer, e.g., to \cite[\S~7]{Anthoine:hal-02875379} for a more
detailed comparison between PoR and PDP schemes.
In fact, a lower bound argument from
\cite[Theorem~4]{Anthoine:hal-02875379} proves that some time/space
tradeoff is inherent. Roughly speaking, for any PoR on an $N$-bit
database, the product of persistent storage overhead times audit
computational complexity must be at least $N$.
We show in~\cref{sec:por} that with VESPo, we let the Server perform
most of the computations (but this remains quite fast), so that we
still need only negligible extra Server storage, but drastically
reduce the Client communication, storage and computations.

\section{Security properties and assumptions}\label{sec:secu}
\subsection{Preliminaries}\label{ssec:prelim}
\paragraph{Pairings.} In the following, we use the notation
$e:\GG_1{\times}\GG_2\rightarrow\GG_T$ to denote a
bilinear pairing in groups of the same prime order.
If such a pairing exists then $\GG_1$ and $\GG_2$ are denoted as
bilinear groups.
We often use groups of prime order, in order to be compute within the exponents.
In particular, thanks to the homomorphic property
of exponentiation, we will perform some linear algebra over the group
and need notations for this.
For a matrix $\AAA$, $g^\AAA$ denotes the coefficient-wise
exponentiation of a generator $g$ to each entry in $\AAA$.
Similarly, for a matrix $\WW$ of group elements and a matrix $\BB$ of
scalars, $\WW^\BB$ denotes the extension of matrix multiplication using
the group action. If we have $\WW=g^\AAA$, then $\WW^\BB =
(g^\AAA)^\BB$.
Further, this quantity can actually be computed if needed by working in the
exponents first, i.e., it is equal to $g^{(\AAA\BB)}$. For example:
\begin{equation}\label{eq:intheexponent}
  \left(g^{\begin{smatrix}a&b\\c&d\end{smatrix}}\right)^{\begin{smatrix}e\\f\end{smatrix}}
  =
  \begin{smatrix}g^a&g^b\\g^c&g^d\end{smatrix}^{\begin{smatrix}e\\f\end{smatrix}}
  =
  \begin{smatrix}g^{ae+bf}\\g^{ce+df}\end{smatrix}
  =
  g^{\left(\begin{smatrix}a&b\\c&d\end{smatrix}\begin{smatrix}e\\f\end{smatrix}\right)}.
\end{equation}

For the sake of simplicity, when there is no ambiguity, we also use
the associated notation shortcuts like:
\(e(g_1^{\begin{smatrix}a\\b\end{smatrix}};g_2^c)=e(g_1;g_2)^{\begin{smatrix}ca\\cb\end{smatrix}}\).

\paragraph{Linearly Homomorphic Encryption (LHE).}
We will also use a public-key partially homomorphic encryption scheme
where both addition and multiplication are considered.
We need the following properties on the linearly homomorphic
encryption function $E$ (according to the context, we use $E_{\pk}$ or
just $E$ to denote the encryption function, similarly for the
decryption function, $D$ or $D_{\sk}$): computing several homomorphic
additions on ciphered messages and homomorphic multiplications but only
between a ciphered message and a cleartext.
\begin{equation}\label{eq:homo:addmul}
 D(E(m_1)E(m_2))= m_1 + m_2
\quad\text{\textbf{AND}}\quad
D(E(m_1)^{m_2}) =  m_1m_2
\end{equation}

\begin{remark}\label[remark]{rk:paillierdp}
For instance, Paillier-like
cryptosystems~\cite{Paillier:1999:homomorph,Benaloh94,Fousse:2011:benaloh}
can satisfy these requirements, \emph{via} multiplication in the
ground ring, for addition of enciphered messages,
and \emph{via} exponentiation for ciphered multiplication.

Note though that an implementation with Paillier cryptosystem of the
evaluation $P(r)$, in a modular ring $\Z_m$, providing the
functionalities of \Cref{eq:homo:addmul}, requires some care:
indeed these equations are usually satisfied modulo an RSA composite
number~$N$, not equal to~$m$.
More precisely, Paillier cryptosystem will provide
$D(E(P(r)))\equiv(\sum_{i=0}^{d}p_i r^i)\mod{N}$. Thus a possibility to
recover the correct value, is to precompute $r^i\mod{m}$ and require
that:
\(
(d+1)(m-1)^2 < N
\).
This way one can actually homomorphically compute over $\Z$ and
use the modulo $m$ only after decryption. See \cref{app:paillier} for
more details.
\end{remark}

\paragraph{Merkle Hash Trees.}
Finally, we will use a Merkle hash tree to allow verifications of
updates.
A Merkle hash tree is a tree in which every leaf is labelled with the
cryptographic hash of a data block, and every other node is labelled
with the hash of the labels of its child
nodes~\cite{Merkle,rfc6962,Anthoine:hal-02875379}.
By just storing the root of the tree, one can check the presence of
a given leaf in the tree with only a logarithmic number of additional
nodes (\emph{uncles}) and hash computations. More details on how we
use them are given in~\cref{ssec:merkle}.

\subsection{Verifiable scheme}\label{ssec:vdpe}
A verifiable dynamic polynomial evaluation (VDPE) scheme consists of five
algorithms: \Setup, \Read, \Update, \Eval and \Verif between a Client
\client{} with state \clstate{}, a Server \server{} with state
\servstate{} and a Verifier \verifier{} with (potentially public)
state \verstate. The algorithms can \reject, when specified, if some
inconsistencies are detected.
\begin{itemize}[leftmargin=2\labelsep]
\item $(\clstate,\verstate,\servstate)\gets\Setup(1^\compsec,P)$:
  On input of the security parameter $\compsec$ and the polynomial
  $P$ of degree $d^{\circ}(P)=d$,
  outputs the Client state $\clstate$, the Verifier \verstate{} and
  the Server state $\servstate$.
  We denote by $\SetupSet_{\compsec}(P)=\{\Setup(1^\compsec,P)\}$ the
  set of admissible states for a given polynomial (dependent on the
  different random choices).
\item $\{p_i,\reject\}\gets\Read(i,\verstate,\servstate)$: On input of an index $i\in{}0..d$, the
  Verifier/Server states \verstate/\servstate, outputs the $i^{th}$
  coefficient of $P$ or \reject.
\item
  $\{(\clstate',\verstate',\servstate'),\reject\}\gets\Update(i,p'_i,\clstate,\verstate,\servstate)$:
  On input of an index $i\in{}0..d$, data $p'_i$, the
  Client/Verifier/Server states \clstate/\verstate/\servstate, outputs
  new Client/Verifier/Server states \clstate'/ \verstate'/ \servstate',
  representing the polynomial $P + (p'_i-p_i) X^i$, or \reject.
  A variant of this algorithm, $\deltaUpdate$, takes as input
  the difference data $\delta=p'_i-p_i$ instead of $p'_i$.
\item $\{\zeta, \bar{\xi} \}\gets\Eval(\servstate,r)$:
  On input of the Server state $\servstate$
  and an evaluation point $r$, outputs $\zeta$ the encrypted value of
  $P(r)$ and a proof $\bar{\xi}$.
\item $\{z,\reject\}\gets\Verif(\verstate,r, \zeta, \bar{\xi})$:
  On input of the Verifier state $\verstate$, the evaluation point
  $r$, the encrypted value $\zeta$ of $P(r)$ and the proof
  $\bar{\xi}$, outputs a successful evaluation \(z=P(r)\) or \(\reject\).
\end{itemize}
The Client may use random coins for any algorithm.
This is the general setting for \emph{public verification}, the idea
being that for a \emph{private verification}, the Client will play the
role of the Verifier too and their states will be identical:
$\verstate=\clstate$.
\subsection{Security properties}\label{ssec:secprop}
Adapted from both~\cite{Kate:2010:KZG,Fiore:2014:CCS}, in order to
take into account dynamicity, we propose the following security game
and the associated security properties:
\begin{figure}[htbp]%
\caption{VDPE soundness security game between two
    Observers $\observer_1$ \& $\observer_2$ (respectively playing the
    roles of the Client and the Verifier), a potentially {\em
      malicious} Server $\adversary$ and an honest Server
    $\server$}\label{game:sound}
\fbox{
\begin{minipage}{.95\columnwidth}
\begin{enumerate}[leftmargin=3\labelsep]
  \item $\adversary$ chooses an initial polynomial $P$.
  \item\label{game:init} $\observer_1$ runs \Setup{},
    keeps $\clstate$ and sends the initial Server part,
    $\servstate$, of the memory layout to both~$\adversary$ and~$\server$;
    and the Verifier part, $\verstate$, to $\observer_2$.
  \item\label{game:poly} For a polynomial number of steps
    $t = 1, 2, . . . , poly(\compsec)$,
    $\adversary$ picks an operation $op_t$ where operation $op_t$ is either
    \Update{}, \Read, \Eval or \Verif. $\observer_1$ executes the
    \Update{} operations with both~$\adversary$ and~$\server$, while
    $\observer_2$ executes the \Read, \Eval{} or \Verif operations
    also with both~$\adversary$ and~$\server$.
  \item $\adversary$ is said to win the game, if
    any cleartext sent by~$\adversary$ differs from that of~$\server$,
    or, if any ciphered message sent by~$\adversary$ does not
    deciphers like that of~$\server$,
    and neither~$\observer_1$ nor~$\observer_2$ did witness \reject{}.
  \end{enumerate}
\end{minipage}
}
\end{figure}

\begin{definition}\label[definition]{def:secVDPE}
  (\Setup, \Read, \Update, \Eval, \Verif) is a secure publicly
  verifiable polynomial evaluation scheme if it satisfies the
  following three properties:
\begin{enumerate}[leftmargin=2\labelsep,label={(\roman*)}]
\item {\bf Correctness.} Let \(d\in\N\),
    \(P(X)=\sum_{i=0}^d p_iX^i\), with \((p_0,\ldots,p_d)\) in a ring
    \R and \((\clstate,\verstate,\servstate)\in\SetupSet_{\compsec}(P)\),
    then for any \(0\leq{i}\leq{d}\) and \(p'_i \in\R\):
\begin{align}
\Read(i,\verstate,\servstate)&=p_i\\
\Update(i,p'_i, \clstate,\verstate,\servstate)&\in\SetupSet_{\compsec}(P+(p'_i-p_i)X^i) \\
\Verif(\verstate,r,\Eval(\servstate,r))&=P(r)
\end{align}
\item {\bf  Soundness.}
    The soundness requirement stipulates that the Client (or the Verifier)
    always \reject (except with negligible probability) if any message sent
    by the Server deviates from the honest (correct)
    behavior\footnote{One might also ask for \emph{knowledge
        soundness} (see, e.g.,~\cite{Lindell:2003:stpc}) on the
      coefficients of the outsourced polynomial, but this is easily
      achieved for any VDPE scheme by definition: the Client can
      simply interpolate from the evaluations.}.
    A VDPE scheme is \emph{sound}, if no polynomial-time
    adversary has more
    than negligible probability in winning the security game
    of~\cref{game:sound}.
\item {\bf  Privacy.} We use now the following variant game:
$\adversary$ chooses two initial polynomials $P_0$, $P_1$;
$\observer_1$ randomly chooses one bit $b\in\{0,1\}$; the players run
steps~(\ref{game:init}-\ref{game:poly}) of~\cref{game:sound} on
$P_b$.
A VDPE scheme is {\em private}, if no polynomial-time adversary has
more than negligible probability in obtaining $b$.
\end{enumerate}
\end{definition}

\begin{definition}\label[definition]{def:secPrivVDPE} (\Setup, \Read, \Update, \Eval, \Verif) is a
  secure \emph{privately} verifiable polynomial evaluation scheme if
  it verifies the
  \textbf{Correctness}, \textbf{Soundness} and \textbf{Privacy}
  requirements of~\cref{def:secVDPE},
  where the Verifier
  state \verstate is included in the Client state \clstate and no
  polynomial-time adversary~$\adversary$ has more than negligible
  probability in winning either the soundness or the privacy security
  games when~$\observer_1$ also plays the role of~$\observer_2$.
\end{definition}

In \cref{sec:full} we apply our new verifiable protocols
to create a new Dynamic Proofs of Retrievability (DPoR)
scheme, provably achieving correctness, soundness, and retrievability
for DPoR. We follow the exact same security definition for DPoR as in
\cite{Anthoine:hal-02875379}, adapted from~\cite{Shi:2013:orampor},
which we will not restate here for the sake of brevity.

\subsection{Assumptions}\label{ssec:assumptions}
To prove the security of our protocols we rely on classical discrete
logarithm and Diffie-Hellman like assumptions, all related to
polynomial computations.
The first assumption, a decisional one, is the distinct leading
monomials assumption: informally it states that
polynomial evaluations ``in the exponents'' where the polynomials have
distinct leading monomials are merely indistinguishable from
randomness. The formal version is recalled in~\cref{def:DLM}.
We also need computational assumptions, including the hardness to
compute discrete logarithms, in~\cref{def:DLOG}, and polynomial extensions
of the hardness to produce Diffie-Hellman-like
secrets even with bilinear pairings, in~\cref{def:tBSDH}.

\newcommand{\citekatzdef}{\cite[Def.~9.63]{Katz:2020:introduction}}
\begin{definition}[Discrete Logarithm, \textbf{DLOG}, hardness
  assumption~\citekatzdef]\label[definition]{def:DLOG}
For a computational security parameter $\compsec\in\N$, a
discrete-logarithm problem is hard relatively to a group \GG of group
order $p\geq{2^{2\compsec}}$, a
generator $g$ and a randomly sampled element of the group,
$h\random{\GG}$, if for
any probabilistic polynomial-time (ppt) algorithms ${\mathcal{A}}_{DLOG}$, there
exists a negligible function $\negl$ such that
\(\Proba{\mathcal{A}_{DLOG}(\GG,g,h)=x\text{~s.t.~}h=g^x}\leq\negl(\compsec)\).
\end{definition}

\begin{definition}[t-Bilinear Strong Diffie-Hellman, \textbf{t-BSDH},
  assumption, from~\cite{Goyal:2007:tBSDH,Kate:2010:KZG}]\label[definition]{def:tBSDH}
For a computational security parameter $\compsec\in\N$,
let $\alpha\in\Z_p^*$, with $p\geq{2^{2\compsec}}$, and $j\in\{1,2\}$. Given as input a $(t + 1)$-tuple
$\left\langle{}g_j,g_j^\alpha,g_j^{\alpha^2},\ldots,g_j^{\alpha^t}\right\rangle\in\GG_j^{t+1}$,
in a bilinear group $\GG_j$ of order $p$ with a bilinear pairing
$e:\GG_1{\times}\GG_2\rightarrow\GG_T$, for every ppt-adversary
$\mathcal{A}_{t-BSDH}$ and for any value of
$c\in\Z_p\backslash\{-\alpha\}$, we have the probability:
{\small\[\Proba{\mathcal{A}_{t-BSDH}(g_1,g_2,g_j^\alpha,g_j^{\alpha^2},\ldots,g_j^{\alpha^t})=\left\langle{}c,\gen^{\frac{1}{\alpha+c}}\right\rangle}\leq\negl(\compsec)\]}
\end{definition}

Next is the distinct leading monomial (DLM) assumption that states
that polynomial evaluations
``in the exponents'' where the polynomials have distinct leading
monomials are merely indistinguishable from randomness.
In \cite{Abdalla:2015:PRF} the assumption is given for $n$-multi\-variate
polynomials with matrices of dimension $k{\times}k$ and projections of
dimension $k{\times}m$ for $k\geq{2}$ and $m\geq{1}$. Here, we will
only use univariate polynomials, $n=1$, and dimensions $k=2$, $m=1$. We
therefore recall the assumption only for this particular case.

\newcommand{\citeabdallathm}{\cite[Theorem~6]{Abdalla:2015:PRF}}
\begin{definition}[Distinct Leading Monomial, \textbf{DLM},
  assumption~\citeabdallathm]\label[definition]{def:DLM}
Let $\GG =\langle{g}\rangle$ be a bilinear group of prime order $p$. The
advantage of an adversary $\mathcal{A}$ against the $(2,1,d)$-DLM
security of $\GG$, denoted
$\text{Adv}_\GG^{(2,1,d)-\text{DLM}}(\mathcal{A})$,
is the probability of success in the game defined in
\cref{game:DLM} and is negligible, with $\mathcal{A}$ being restricted to
make queries $P\in{Z_p[T]}$ such that for any challenge
$P$, the maximum degree in one indeterminate in $P$ is at most $d$,
and for any sequence $(P_1,\ldots,P_q)$ of queries, there exists an
invertible matrix $M\in\Z_p^{q{\times}q}$ such that the leading
monomials of $M\cdot{}\Transpose{[P_1,\ldots,P_q]}$ are distinct.
\begin{table}[!ht]\centering\caption{$(2,1,d)$-DLM
    security game for a bilinear group
    $\GG$~\cite{Abdalla:2015:PRF}}\label{game:DLM}
\setlength{\tabcolsep}{8pt}\renewcommand{\arraystretch}{1.5}
\begin{tabular}{lll}
\toprule
\multicolumn{1}{c}{Init} & \multicolumn{1}{c}{Challenge$(P)$} & \multicolumn{1}{c}{Response$(b')$}\\
\midrule
$\vect{r}\random\Z_p^{2{\times}2}$ & \textbf{If} $b==0$ & \multirow{3}{*}{\emph{Return} $b'==b$}\\
$\beta\random\Z_p^{2}$ & \textbf{Then} \emph{Return} $y\gets{g^{P(\vect{r})\cdot\beta}}$&\\
$b\random\{0,1\}$ & \textbf{Else} \emph{Return} $y\random\GG^{2}$&\\
\bottomrule
\end{tabular}
\end{table}
\end{definition}

In fact, the DLM security can also be reduced to the Matrix Diffie-Hellman
assumption (MDDH)~\cite[Theorem~5]{Abdalla:2015:PRF}, a generalization of the
widely used decision linear
assumption~\cite{Gay:2016:eurocrypt,Morillo:2016:KMDDH,Agrawal:2016:crypto,Ambrona:2017:crypto,Attrapadung:2020:asiacrypt}.

Eventually, when we use Merkle Hash Tree, we need to apply a
\emph{Collision Resistant Hash Function} (\textbf{CRHF}), so that
finding different hash trees with the same root is hard.

Overall, since we consider the semantic security of the
cryptosystem, we assume that adversaries are probabilistic polynomial
time machines. More precisely we consider \textbf{Malicious
  adversaries}: a corrupted Server controls the network and
stops, forges or listens to messages in order to gain information or
fool the Client.

\section{Tools for the verification of a polynomial evaluation}\label{sec:cipher}

Our first step is to define a verification protocol for
polynomial evaluation that supports a ciphered input polynomial over a
finite ring $\Z_p$.
For this we start with ideas mostly
from~\cite{Kate:2010:KZG,Benabbas2011VerifiableDO}, in order to
highlight the difficulties in our setting: adding dynamicity and
encryption; that is allowing to modify only parts of the inputs at a
low cost, while dealing with covert inputs and preserving a proven
security.
More precisely, our modifications allow the adaptation of the security
proof in order to incorporate the updates, and require some
algorithmic tricks to preserve the linearity of the Server
computations.

First, we define a difference
polynomial that we will use to check consistency.

\begin{definition}
For a polynomial $P(X)\in\Z_p[X]=\sum_{i=0}^d p_i X^i$ of degree $d$,
let its \emph{subset polynomials} be: $T_{k,P}(X)=\sum_{i=k+1}^d
p_{i} X^{i-k-1}=\sum_{j=0}^{d-1-k} p_{j+k-1} X^{j}$.
\end{definition}
\begin{lemma}\label[lemma]{prop:bivariate}
Let $Q_P(Y,X)=\frac{P(Y)-P(X)}{Y-X}$ be the \emph{difference
  polynomial} of a polynomial $P$;
then:
\begin{equation}\label{eq:diffpoly}
Q_P(Y,X)=\sum_{i=1}^d p_i \sum_{k=0}^{i-1} Y^{i-k-1}X^k=\sum_{k=0}^{d-1} T_{k,P}(Y) X^k\end{equation}
\end{lemma}
\begin{proof}
As $Y^i-X^i=(Y-X)(\sum_{k=0}^{i-1} Y^{i-k-1}X^k)$, we obtain that
\(Q_P(Y,X)=\sum_{i=1}^d p_i \sum_{k=0}^{i-1} Y^{i-k-1}X^k\).
This is also
\(Q_P(Y,X)=\sum_{k=0}^{d-1} X^k\left(\sum_{i=k+1}^d p_iY^{i-k-1}\right).
\)
\myqed\end{proof}

This identity relates two evaluations of $P$:
$P(Y)=P(X)+(Y-X)Q_P(Y,X)$. This equation allows one to verify
$z\checks{=}P(r)$ by checking, for a secret $s$, that:
\begin{equation}\label{eq:KZG}
P(s)=z+(s-r)Q_P(s,r)
\end{equation}
\begin{table*}[!ht]\centering
  \caption{Verifiable Ciphered Polynomial
    Evaluation}\label{proto:cKZG}
\fbox{\adjustbox{max width=.975\textwidth}{\centering
    \begin{tabular}{cccc}
      & Server & Communications & Client \\
      \midrule
      \multirow{3}{*}{{\Setup}}&&$\GG_1,\GG_2,\GG_T$ groups of order
      $p$& $P\in\Z_p[X]$, $1\leq{}d^{\circ}(P)\leq{d}$; let $s\random\Z_p$ \\
      &&pairing $e$, gen. $g_1,g_2,g_T=\gen$& $W\gets{E_{\pk}(P)}$,
      $\mathcal{K} \gets g_T^{P(s)}$, $H\gets[g_2^{T_{k,P}(s)}]_{k=0}^{d-1}$\\
& {\bf{Output}} : $\servstate=\{\pk, \GG_{1,2,T},W, H\}$ &\pleftarrow{W,H}&  $\clstate=\{\pk,\sk,\GG_{1,2,T},g_{1,2,T},e,\mathcal{K},s\}$\\
      \midrule
      \multirow{2}{*}{\Eval{}/\Verif{}}& Form
      $x\gets\Transpose{[1,r,r^2,\ldots,r^d]}$ & \pleftarrow{r} & $r\random\Z_p$\\
      &$\zeta=\Transpose{W}\boxdot{x}$;
      $\xi=\Transpose{H}\odot{x_{0..d-1}}$& \prightarrow{\zeta,\xi} &$e(g_1^{s-r};\xi)g_T^{D_{\sk}(\zeta)}\checks{=}\mathcal{K}$\\
&&& {\bf{Output}} : $D_{\sk}(\zeta)$ or {\bf{reject}}\\
    \end{tabular}
}}
\end{table*}

For this, let $E,D$ be the encryption and decryption functions of a
partially homomorphic cryptosystem, supporting addition of two ciphertexts
and multiplication of ciphertext by a cleartext, as
in~\Cref{eq:homo:addmul}. Therefore it is possible to evaluate a
ciphered polynomial at a clear evaluation point, using powers of the
evaluation point:
for $x=[1,r,r^2,\ldots,r^d]$, denote by
$\Transpose{E(P)}\boxdot{x}=\prod_{i=0}^dE(p_i)^{r^i}=E(P(r))$, the
homomorphic polynomial evaluation.

Similarly, if $H=[h_0,\ldots,h_d]=[g^{a_0},\ldots,g^{a_d}]$,
denote by $H\odot{x}=\prod_{i=0}^dh_i^{x_i} = g^{\sum{a_ix_i}}$ the
dot-product in the exponents. Then~\cref{proto:cKZG} shows how the
Server produces the evaluation via the partially homomorphic cipher
and the subset polynomials (in this table, and in the following
protocols presentations, time passes from top to bottom only, driven
by the "Communications" column).
Then this evaluation is bound to be correct by the consistency
check in the exponents.

\begin{restatablebackref}%
{propositionS}{cKZGverifiable}{prop:cKZGverifiable}{app:proofs}
The protocol of~\cref{proto:cKZG} is correct and sound under the
$d$-BSDH assumption.
\end{restatablebackref}
Several issues remain with this protocol: first it is not dynamic.
Indeed, for a dynamic version, the problem is that updating only one
coefficient of $P$ requires to update up to $d-1$ coefficients
of~$H$. This work would be of the same order of magnitude as recomputing
the whole setup.
Second it is not fully hiding the coefficients of $P$ as they are
just put in the exponents without any masking, and we do not prove the
privacy requirement\footnote{Efficient updates in similar schemes are
  considered, e.g., in~\cite{Tomescu:2020:aggregatable} but to a protocol
  that verifies coefficients known to the Server, \emph{not} its
  evaluation at hidden coefficients}.
Third, the protocol is not fully publicly verifiable since the
decryption key of the partially homomorphic system is required.
We incrementally solve the first two issues in the remainder of this
paper and obtain a fully secure private protocol.
We also are able to provide a dynamic protocol, publicly verifiable,
but for an unciphered polynomial.
Combining all three properties, that is,
designing a publicly verifiable dynamic protocol for ciphered
polynomials, preserving a good efficiency while still being secure,
remains an open question to us
(usually when adapting a static protocol, either dynamicity involves
too much recomputation or the security is compromised by the updates).
\section{Outsourced dynamic verification of the
  evaluation}\label{sec:dynamic}
In order to be able to deal with updates, a classical tool is to add
Merkle trees that are updated along with the polynomial parts.
Checking the root of the Merkle tree allows for logarithmic
verifications and updates of any coefficient of the polynomial.
Modifications of the polynomial coefficients are also included in the
Client state so that old polynomials cannot be used for the
verification of \Eval.
The difficulty then is to preserve a linear time Server with a fast and
light Client; we show next how to achieve this.

\subsection{Merkle trees for logarithmic Client storage}\label{ssec:merkle}
In order to avoid storing the polynomial coefficients on the
Client side,
we thus use a Merkle hash tree~\cite{Merkle,rfc6962,Anthoine:hal-02875379}. %
Then it is sufficient to store
the root of the Merkle tree:
under the CRHF assumption, a malicious Server
cannot give back different polynomial coefficients.
For our purpose, an implementation of such trees must just provide
the following algorithms:

\begin{itemize}[leftmargin=2\labelsep]
\item $T\leftarrow\mtcreate(X)$ creates a Merkle hash tree from a
  database~$X$.
\item $r\leftarrow\mtrootfromleaves(X)$ computes from scratch the root
  of the Merkle hash tree of the whole database~$X$.
\item $L\leftarrow\mtuncles(i,T)$ retrieves a list of ``uncle'' node
  hashes along the path in the tree to index $i$.
\item $r\leftarrow\mtrootfrompath(i,a,L)$ computes the root of the
  Merkle hash tree from a leaf element~$a$ and the associated path of
  uncles~$L$.
\item $T'\leftarrow\mtupdateleaf(i,a,T)$ updates the whole Merkle
  tree~$T$ by changing the $i$-th leaf to be~$a$.
\end{itemize}

The correctness requirements are that,
for any index $i$ and databases $X,Y$ which are identical except possibly
for the $i$'th index index (i.e., $\forall j\ne i, x_j = y_j$), we have

\begin{eqnarray}
\mtrootfromleaves(X) &=&
  \mtrootfrompath\left(i, x_i, \mtuncles\left(i,\mtcreate(X)\right)\right)
  \label{eq:merkle}
  \\
  \mtcreate(X) &=& \mtupdateleaf(i,x_i,Y)
  \label{eq:mtupl}
\end{eqnarray}

And the soundness requirement is that no P.P.T.\ adversary can compute
a tuple $(X,i,b,L)$ such that
\begin{equation}
x_i \ne b \quad\text{and}\quad
\mtrootfromleaves(X) = \mtrootfrompath(i, b, L).
\label{eq:mtsound}
\end{equation}

\subsection{Public dynamic unciphered polynomial evaluation}
Thanks to these additional Merkle-tree operations, we can now give a
protocol for the public verification of the evaluation of a dynamic
polynomial $P$.
It consists in five algorithms (\Setup, \Read, \Update, \Eval, \Verif) detailed
in~\cref{protoDynClear} and it requires, for now, a \emph{symmetric} pairing.

\begin{table*}\centering\renewcommand{\arraystretch}{1.1}
  \caption{Public and  Dynamic unciphered polynomial evaluation}\label{protoDynClear}
  \adjustbox{max width=\textwidth}{\centering
\fbox{\begin{tabular}{cccc}
    & Server & Communications & Client/Verifier \\
    \midrule
    \multirow{4}{*}{{\Setup}}&&$\GG,\GG_T$ of order $p$, gen. $g$&
    $P\in\Z_p[X]$, $1\leq{}d^{\circ}(P)\leq{d}$; let $s\random\Z_p$ \\
    &&symm. pairing $e$& $\mathcal{K}_1\gets e\left(g^{{P}(s)};g\right)$,
    $\mathcal{K}_2 \gets g^s$, $S \gets [g^{s^k}]_{k=0}^{d-1}$\\
    &$T_P\gets\mtcreate(P)$ & \pleftarrow{P,S} & $r_P\gets \mtrootfromleaves (P)$\\
& {\bf{Output}} : $\servstate=\{\GG,P,T_P,S\}$ &&  $\verstate=\{\GG,\GG_T,g,e,\mathcal{K}_1,\mathcal{K}_2,r_P\}$,
$\clstate=\verstate\cup\{s\}$ \\
    \midrule
    \multirow{2}{*}{{\Read}} && \pleftarrow{i} &\\
    &$L_i\gets\mtuncles(i,P,T_P)$ &  \prightarrow{p_i,L_i}
    & $r_P\checks{=} \mtrootfrompath(i,p_i,L_i)$\\
&&& {\bf{Output}} : $p_i$ or {\bf{reject}}\\
    \midrule
    \multirow{4}{*}{{\Update}} && \pleftarrow{i,p'_i} &\\
    &$L_i\gets\mtuncles(i,P,T_P)$ &
    \prightarrow{p_i,L_i} & $r_P\checks{=}\mtrootfrompath(i,p_i,L_i)$,
    $r'_P\gets \mtrootfrompath(i,p'_i,L_i)$\\
    &  $T_P \gets \mtupdateleaf(i,p'_i,T_P)$; $p_i \gets p_i'$  &&
    $\mathcal{K}'_1\gets\mathcal{K}_1\cdot{e\left(g^{s^i(p'_i-p_i)};g\right)}$\\
& {\bf{Output}} : $\servstate=\{\GG,P,T_P,S\}$ &&  $\verstate=\{\GG,\GG_T,g,e,\mathcal{K'}_1,\mathcal{K}_2,r'_P\}$,
$\clstate=\verstate\cup\{s\}$ or {\bf{reject}}\\
    \midrule
    \multirow{2}{*}{\Eval{}/\Verif{}}& Form
    $x\gets\Transpose{[1,r,r^2,\ldots,r^d]}$ & \pleftarrow{r} &
    $r\random\Z_p$\\
    &$\zeta \gets P(r)$;
    ${\xi}\gets\prod_{i=1}^d\prod_{k=0}^{i-1}S_{i-k-1}^{p_ix_k}$&
    \prightarrow{\zeta,\xi} &
    $e(\xi;\mathcal{K}_2/g^r)e(g^{\zeta};g )\checks{=}\mathcal{K}_1$\\
&&& {\bf{Output}} : $D(\zeta)$ or {\bf{reject}}\\
  \end{tabular}
}}
\end{table*}

During the \Setup~algorithm, the Client sends the unciphered polynomial
to the Server and  deletes it  to minimize its storage. The
Client uses a random coin $s$ to create some data to be published or
to be sent to the Server.
The Verifier collects
the published data and is authorized to run the \Read~and the
\Verif~algorithms. But she is not authorized to run the \Setup~and
\Update~algorithms (she does not know $s$).
At any point the Client can take the role of a Verifier.
This is shown in~\cref{protoDynClear}, where the different states are
as follows:
$\verstate=\{\GG,\GG_T,g,e,\mathcal{K}_1,\mathcal{K}_2,r_P\}$,
$\clstate=\verstate\cup\{s\}$ and
$\servstate=\{\GG,g,P,T_P,S\}$.

\begin{restatablebackref}%
{propositionS}{csClear}{prop:csClear}{app:proofs}
The protocol of~\cref{protoDynClear}
  is correct and sound under the $d$-BSDH and CRHF assumptions.
\end{restatablebackref}
One difficulty is to preserve a linear-time Server. We
show next that this is indeed possible here.

\subsection{Efficient linear-time evaluation}\label{ssec:dcqComplex}
As a first approach to evaluate our protocols, we consider that the
cardinality of the coefficient domain is a constant. Therefore,
we count as arithmetic operations in the field not only the usual
addition, subtraction, multiplication and inversion, but also the
exponentiations that are independent of the degree of the polynomial.
We thus express our asymptotic complexity bounds
in~\cref{tab:PVDUeval}, only with respect to that degree $d$.
The main idea is to evaluate the polynomial of~\cref{eq:diffpoly}
(with a priori a quadratic number of monomials) in a Horner-like
fashion, so that it requires only a linear number of operations.
\begin{table}[!ht]\centering
\renewcommand*{\arraystretch}{1.25}
\caption{Complexity bounds for the publicly  verifiable dynamic and unciphered
  polynomial evaluation of~\cref{protoDynClear} for a degree $d$
  polynomial.}\label{tab:PVDUeval}
  \adjustbox{max width=\textwidth}{\centering
    \begin{tabular}{ccccc}
      \toprule
      \multicolumn{2}{c}{}& Server & Communication & Client \\
      \midrule
      \multicolumn{2}{c}{Storage} & \bigO{d} & & \bigO{1} \\
      \midrule
      \multirow{3}{*}{\rotatebox[origin=c]{90}{Comput.}}
      &\Setup	& \bigO{d}	& \bigO{d}	& \bigO{d} \\
      &\Read/\Update	& \bigO{\log(d)}	& \bigO{\log(d)}	& \bigO{\log(d)} \\
      &\Eval/\Verif{}& \bigO{d}	& \bigO{1}	& \bigO{1} \\
      \bottomrule
    \end{tabular}
  }
\end{table}
\begin{restatablebackref}{propositionS}{compClear}{prop:compClear}{app:proofs}
  In~\cref{protoDynClear},
  the setup protocol requires \bigO{d}
  arithmetic and hashing operations;
  the update protocol requires \bigO{log(d)}
  arithmetic and hashing operations;
  the verification protocol requires
  \bigO{1} communications and arithmetic operations for the Client,
  and \bigO{d} arithmetic operations for the Server.
\end{restatablebackref}
In the next Section, we then propose a novel fully private protocol,
combining and formalizing the ideas from the encrypted one and the
dynamic one.

\section{Fully private, dynamic and ciphered polynomial
  evaluation}\label{sec:full}
So far we have a polynomial evaluation verification, that allows
efficient updates of its
coefficients. We now propose a scheme which combines the polynomial
evaluation with the  externalization of  the polynomial itself. For
this, two more ingredients are added in~\cref{ssec:full}: an efficient
masking in the exponents in order to fulfill the hiding security
property and an outsourcing of the (ciphered) polynomial itself.
This latter feature allows the Client to not even store the
polynomial and reduces her need for permanent storage to a small constant
number of field elements. For this we use Merkle hash trees presented
in~\cref{ssec:merkle}. They ensure the authenticity of the coefficient
updates, with the storage of only one hash. Finally note that the
bilinear pairing need not be symmetric anymore, but need to be applied
twice for the security hypothesis to hold.

We start the section with the security tools and then the
linear algebra algorithms we will use and
then give a full formalization and the associated proofs of
security. We end the section with experiments showing the efficiency
of our approach.

\subsection{Security requirements}\label{ssec:full}
Here we add a secret masking of the polynomial coefficients in order to make
the protocol hiding.
For this we use the security hypothesis of~\cref{def:DLM}:
indeed, DLM security states that in a bilinear group $\GG$ of prime order, the
values $(g^{P_1(A)\vect{\beta}},\ldots,g^{P_d(A)\vect{\beta}})$ are indistinguishable from a
random tuple of the same size, when $P_1,\ldots,P_d$ have distinct
leading monomials of bounded degree and $A$ and $\vect{\beta}$
are the $2{\times}2$ and $2{\times}1$ secrets.
Therefore, in our modified protocol, the coefficients
$g^{\Phi^i\vect{\beta}}$ for a secret $2{\times}2$ matrix $\Phi$, are
indistinguishable from a random tuple ($g^{\Gamma_i}$) since the
polynomials $X^i$, $i=1..d$ are just distinct monomials.

\subsection{Linear algebra toolbox}
For the next protocol to hold, we need to adapt the difference polynomial
to the matrix case. For instance~\cref{prop:bivariate} holds in the
matrix case provided that the, now matrices, $Y$ and $X$ commute and
that $Y-X$ is invertible. Let $I_n$ be the $n{\times}n$ identity
matrix. Then, we will for instance use $Y=sI_2$ and $X=rI_2$ with $s\neq{r}$.

Also to speed-up things with the DLM masks, we need to efficiently
compute geometric sums of matrices.
Thanks to Fiduccia's algorithm~\cite{Fiduccia:1985:linrec},
this is easily done with a number of operations logarithmic in the
exponent, provided that $1$ is not an eigenvalue of the
matrix. Indeed, first, any matrix commutes with the identity so the
geometric sum can be computed via one matrix exponentiation, one
matrix inverse and one matrix multiplication:
\(
\sum_{i=0}^d A^i = (A^{d+1}-I_n)(A-I_n)^{-1}
\).
Then, second, Fiduccia's algorithm computes the exponentiation modulo
the characteristic polynomial, using the square and multiply fast
recursive algorithm. This is summarized
in~\cref{alg:fiduccia,alg:matgeomsum} and analyzed in~\cref{app:proofs}.

\begin{algorithm}[!ht]
  \caption{Degree $2$ modular monomial powers ($2$-MMP)}\label{alg:fiduccia}
  \begin{algorithmic}[1]
    \REQUIRE $d\in\Z$, $d\geq{1}$, $P=p_0+p_1Z+Z^2\in\Z_p[Z]$ monic degree $2$ polynomial.
    \ENSURE $Z^d\mod{P}$.
     \IfThen{$d==1$}{\algorithmicreturn{} $Z$}
     \STATE $T\gets{2\text{-MMP}(\lfloor{d/2}\rfloor,P)}$;
     \STATE $S\gets (t_0^2-t_1^2p_0)+(2t_0t_1-t_1^2p_1)Z$;
     \hfill\COMMENT{$T(Z)^2$ modulo $P(Z)$}
     \IF{$d$ is odd}
     \RETURN $(-s_1p_0)+(s_0-s_1p_1)Z$;
     \hfill\COMMENT{$Z\cdot{S(Z)}$  modulo $P(Z)$}
     \ELSE
     \RETURN $S$.
     \ENDIF
  \end{algorithmic}
\end{algorithm}

\begin{algorithm}[!ht]
  \caption{Projected matrix geometric sum (PMGS)}\label{alg:matgeomsum}
  \begin{algorithmic}[1]
    \REQUIRE $k\in\Z$,
    $A=\begin{smatrix}a&b\\c&d\end{smatrix}\in\Z_p^{2{\times}2}$,
    s.t. $A-I_2$ is invertible, $\vect{\beta}\in\Z_p^2$.
    \ENSURE $\sum_{i=0}^k A^i\vect{\beta}$.
     \STATE Let
     $\pi(Z)=(ad-bc)-(a+d)Z+Z^2$;
     \hfill\COMMENT{The characteristic polynomial of $A$}
     \STATE Let $F(Z)=f_0+f_1Z=2\text{-MMP}(k+1,\pi)$; \hfill\COMMENT{$Z^{k+1}\mod\pi(Z)$, using~\cref{alg:fiduccia}}
    \RETURN $(f_1 A+(f_0-1)I_2)(A-I_2)^{-1}\vect{\beta}$.\hfill\COMMENT{$(A^{k+1}-I_2)(A-I_2)^{-1}\vect{\beta}$}
  \end{algorithmic}
\end{algorithm}

\begin{lemma}\label{lem:matfid}
  \Cref{alg:matgeomsum}, computing the matrix geometric sum, requires
  between $40+8\lceil\log_2(d_p+1)\rceil$ and
  $40+11\lceil\log_2(d_p+1)\rceil$ arithmetic operations.
\end{lemma}
\begin{proof}
  Counting only (modular) field operations, \cref{alg:fiduccia}
  requires between $8$ and $11$ times $\lceil\log_2(d_p)\rceil$ additions
  and multiplications depending on
  the binary decomposition of $d_p$.
  Then we have $5$ operations for the matrix inverse, twice $6$ operations for
  the matrix-vector multiplications and $18$ operations for the matrix
  polynomial evaluation. Plus $5$ operations for the characteristic
  polynomial.
\myqed\end{proof}

\subsection{Formalization of the protocol}

The dynamic externalized polynomial evaluation scheme consist of the
following algorithms \Setup, \Read, \Update, \Eval and \Verif between a
Client $\client$ with state \clstate~and the Server $\server$ of state
\servstate.
Following the definition of a VDPE scheme of~\cref{ssec:vdpe},
\Setup is detailed in~\cref{alg:setup};
\Read is detailed in~\cref{alg:read};
\Update is detailed in~\cref{alg:update};
\Eval is detailed in~\cref{alg:veval};
and
\Verif is detailed in~\cref{alg:verif}.
Finally, a lighter variant of \Setup and \Update (detailed in ~\cref{alg:update2}) dedicated to the DPoR protocol is proposed.
The exchanges are summarized in~\cref{proto:full} of~\cref{app:full}.
\begin{algorithm}[!ht]
  \caption{$\Setup(1^\compsec,P)$}\label{alg:setup}
  \begin{algorithmic}[1]
    \REQUIRE $1^\compsec$; $p\in\Primes$, $P=\sum_{i=0}^d p_iX^i\in\Z_p[X]$;
    \REQUIRE a partially homomorphic cryptosystem $E/D$
    satisfying~\cref{eq:homo:addmul}, \emph{for any dot-product of
      size $d+1$, modulo $p$}.
    \ENSURE $\servstate{}$, $\clstate{}$.
     \STATE Client: generates order $p$ groups $\GG_1$, $\GG_2$,
     $\GG_T$ with non-degenerate pairing
     $e:\GG_1\times\GG_2\rightarrow\GG_T$ and generators $g_1,g_2,g_T=\gen$;
     \STATE Client: generates a public/private key pair $(\pk,\sk)$ for $E/D$;
     \STATE Client: randomly selects $s\random\Z_p{\setminus}\{0,1\}$,
     $\vect{\alpha},\vect{\beta}\random\Z_p^2$,
     $\matr{\Phi}\random\Z_p^{2{\times}2}$, s. t. $s\matr{\Phi}-I_2$
	  is invertible;
     \STATE Client: computes
     $\bar{P}(X)=\sum_{i=0}^dX^i({p_i}\vect{\alpha}+\matr{\Phi}^i\vect{\beta})$, $W=E_{\pk}(P)=[E(p_i)]_{i=0}^{d}$,
     $S=[g_1^{s^k}]_{k=0}^{d-1}\in\GG_1^d$,
     $\bar{H}=[g_2^{\bar{p}_i}]_{i=1}^{d}\in\GG_2^{2{\times}d}$,
     $\bar{\mathcal{K}}=g_T^{\bar{P}(s)}\in\GG_T^2$
	  and
	  $d_p{=}d\tmod{\varphi(p)}\equiv{d\tmod{p{-}1}}$;
     \STATE\label{lin:mtWC} Client: $r_W=\mtrootfromleaves (W)$;
     \hfill\COMMENT{root of the Merkle tree}
    \STATE Client: sends $\pk,\GG_1,\GG_2,g_1,g_2,\GG_T,e,W,S,\bar{H}$ to the Server;
    \CRETURN{Client} $\clstate{} \gets\{\pk,\sk,\GG_{1,2,T},g_{1,2,T},e,s,\vect{\alpha},\vect{\beta},\matr{\Phi},\bar{\mathcal{K}},r_W,d_p\}$;
    \STATE\label{lin:mtWS} Server: $T_W\gets\mtcreate(W)$;
    \hfill\COMMENT{\emph{the Merkle tree}}
    \CRETURN{Server} $\servstate{}  \gets\{\pk,\GG_{1,2,T},e,W,T_W,S,\bar{H}\}$.
  \end{algorithmic}
\end{algorithm}

\begin{algorithm}[!ht]
  \caption{$\Read(i \clstate, \servstate)$}\label{alg:read}
  \begin{algorithmic}[1]
    \REQUIRE $i\in [0..d]$, $(\clstate,\servstate)=\Setup(1^\compsec,P)$.
    \ENSURE $p_i$ the value of the $i^{th}$ coefficient of $P$.

\STATE Client: sends $i$;
    \STATE Server: $L_i\gets\mtuncles(i,W,T_W)$;
    \STATE Server: sends $w_i$, $L_i$ to the Client;
    \IF{$r_W\neq\mtrootfrompath(i,w_i,L_i)$}
\newline\COMMENT{\emph{the stored root does not match the received element and
    uncles}}
    \CRETURN{Client} \reject.
    \ELSE

\STATE Client: computes $p_i=D(w_i)$;
    \ENDIF
  \end{algorithmic}
\end{algorithm}

\begin{algorithm}[!ht]
  \caption{$\Update(i, p'_i, \clstate, \servstate)$}\label{alg:update}
  \begin{algorithmic}[1]
    \REQUIRE $i\in [0..d]$, $p'_i\in\Z_p^*$, $(\clstate,\servstate)=\Setup(1^\compsec,P)$.
    \ENSURE $(\clstate',\servstate')=\Setup(1^\compsec,P+(p'_i-p_i) X^i)$ or \reject.
\STATE Client: gets $\clstate=(\pk,\sk,\GG_{1,2,T},g_{1,2,T},e,s,\vect{\alpha},\vect{\beta},\matr{\Phi},\bar{\mathcal{K}},r_W,d_p)$,
\STATE Client: computes $w'_i= E(p'_i)$,
\STATE Client: computes $\bar{H}'_i\gets g_2^{p'_i \vect{\alpha}+\matr{\Phi}^i\vect{\beta}}$
\STATE Client: sends $i,w'_i$ if $(i>0)$ sends $\bar{H}'_i$;

    \STATE Server: $L_i\gets\mtuncles(i,W,T_W)$;
    \STATE Server: $T'_W\gets\mtupdateleaf(i,w'_i,T_W)$;
    \hfill\COMMENT{\emph{updates the Merkle tree}}
    \STATE Server: sends $w_i$, $L_i$ to the Client;
    \STATE Server: $\servstate^*\gets\servstate\backslash\{T_W,w_i\}\bigcup\{T'_W,w'_i \}$
    \IfThenElse{$i==0$}{Server:
      \algorithmicreturn~$\servstate'\gets\servstate^*$\newline}{Server:
      \algorithmicreturn~$\servstate'\gets\servstate^*\backslash\{\bar{H}_i\}\bigcup_{j=1}^2\{\bar{H'}_i[j]\}$}
    \IF{$r_W\neq\mtrootfrompath(i,w_i,L_i)$}
\newline\COMMENT{\emph{the stored root does not match the received element and
    uncles}}
    \CRETURN{Client} \reject.
    \ELSE

\STATE Client: computes $\Delta \gets g_2^{(p'_i-p_i)\alpha}$;
\STATE Client: computes
    $\bar{\mathcal{K}'}[j]\gets{e(g_1; \Delta[j]^{s^i})\cdot\bar{\mathcal{K}}[j]}$ for $j=1..2$;
    \STATE Client: computes $r'_W=\mtrootfrompath(i,w'_i,L_i)$;
    \CRETURN{Client} $\clstate'\gets\clstate\backslash\{\bar{\mathcal{K}},r_W\}\bigcup\{\bar{\mathcal{K}'},r'_W\}$.
    \ENDIF
  \end{algorithmic}
\end{algorithm}

\begin{algorithm}[!ht]
  \caption{$\Eval( \servstate, r)$}\label{alg:veval}
  \begin{algorithmic}[1]
\REQUIRE $ \servstate$ and a evaluation point $r\in\Z_p$;
\ENSURE $\zeta$ the encrypted evaluation of $P(r)$ and a proof $\bar{\xi}$.
\STATE\label{lin:paillier:veval}
Server: computes
$\zeta=\Transpose{W}\boxdot{x}=\prod_{i=0}^{d}w_i^{(r^i\mod{p})}$
\STATEx\COMMENT{\emph{via~\cref{eq:homo:addmul}}, see also, e.g., \cref{rk:paillierdp,alg:paillierdp}}
\STATE Server: $\bar{\xi}=\Transpose{[1_{\GG_T},1_{\GG_T}]}\in\GG_T^2$; $t=1_{\GG_1}$;
\FOR{$i=1$ \TO $d$}\hfill\COMMENT{Following the ideas of~\cref{ssec:dcqComplex}}
\label{lin:begfor:veval}
\STATE Server: $t\gets{S_{i-1}\cdot{t^r}}$;
\STATE Server:
$\bar{\xi}[j]\gets\bar{\xi}[j]\cdot{e(t;\bar{H}_i[j])}$ for $j=1..2$;
\ENDFOR\label{lin:endfor:veval}
\CRETURN{Server}  $\zeta,\bar{\xi}$.
  \end{algorithmic}
\end{algorithm}

\begin{algorithm}[!ht]
  \caption{$\Verif(\clstate, r,\zeta,\bar{\xi})$}\label{alg:verif}
  \begin{algorithmic}[1]
\REQUIRE $\clstate$, the evaluation point $r\in\Z_p$, its encrypted evaluation $\zeta$ and a proof $\bar{\xi}$ ;
\ENSURE $z=P(r)$ or \reject.
\STATE Client: computes $r\matr{\Phi}$ and
$\vect{c}\gets\left((r\matr{\Phi})^{d_p+1}-I_2\right)\cdot(r\matr{\Phi}-I_2)^{-1}\cdot\vect{\beta}$\hfill\COMMENT{via~\cref{alg:matgeomsum}}

\STATE Client: computes $z=D_{\sk}(\zeta)\mod{p}$;
\IF{$\bar{\xi}[j]^{s-r}g_T^{{z}\vect{\alpha}[j]+\vect{c}[j]} =
  \bar{\mathcal{K}}[j]$ for $j=1..2$}
\CRETURN{Client} $z$.
\ELSE
\CRETURN{Client} \reject.
\ENDIF
  \end{algorithmic}
\end{algorithm}

\begin{remark}
We show next how to use a dynamic VPE protocol inside a DPoR scheme.
There, the client updates a polynomial coefficient $p_i$ by sending an
encryption of only the difference $\delta= p'_i-p_i$ without needing
to know the value of $p_i$.
In this variant, the value of $p_i$ does not have to be checked and
the hash tree is superfluous.

We thus consider \deltaSetup, a variant of \Setup where the
client does not need $\mtrootfromleaves(W)$ (line~\ref{lin:mtWC} of
\cref{alg:setup}) and the server does not compute the tree at all
(remove $\mtcreate(W)$, line~\ref{lin:mtWS}). $\clstate{}$ and $\servstate{} $
are therefore reduced to
$\clstate{}=\{\pk,\sk,\GG_{1,2,T},g_{1,2,T},e,s,\vect{\alpha},\vect{\beta},\matr{\Phi},\bar{\mathcal{K}},d_p\}$
and $\servstate{} =\{\pk,\GG_{1,2,T},e,W,S,\bar{H}\}$ (note that this prevents
using the \Read operation on the polynomial).
In addition, we also consider a variant of the \Update algorithm,
which takes $\delta= p'_i -p_i$ as input instead of $p_i$, as detailed
in \cref{alg:update2}. The corresponding dynamic (from the
difference) externalized polynomial evaluation scheme is then reduced
to the algorithms $\deltaSetup$, $\deltaUpdate$, $\Eval$ and $\Verif$.
\end{remark}

\begin{algorithm}[!ht]
  \caption{$\deltaUpdate(i, \delta, \clstate, \servstate)$}\label{alg:update2}
  \begin{algorithmic}[1]
    \REQUIRE $i\in [0..d]$, $\delta\in\Z_p^*$, $(\clstate,\servstate)=\Setup(1^\compsec,P)$.
    \ENSURE $(\clstate',\servstate')=\Setup(1^\compsec,P+\delta X^i)$ or \reject.
    \STATE Client: computes $e_\delta=E_{\pk}(\delta)$, $\Delta=g_2^{\delta\vect{\alpha}}$;
    \STATE Client: sends $i,e_\delta,\Delta$ to the Server;
    \STATE Server: sends $w_i$ to the Client;
    \IfThenElse{$i==0$}{Server:
      \algorithmicreturn~$\servstate'\gets\servstate^*$\newline}{Server:
      \algorithmicreturn~$\servstate'\gets\servstate^*\backslash\{\bar{H}_i\}\bigcup_{j=1}^2\{\bar{H}_i[j]\cdot\Delta[j]\}$}
\STATE Client: computes
    $\bar{\mathcal{K}'}[j]\gets{e(g_1;\Delta[j]^{s^i})\cdot\bar{\mathcal{K}}[j]}$ for $j=1..2$;
    \CRETURN{Client} $\clstate'\gets\clstate\backslash\{\bar{\mathcal{K}}\}\bigcup\{\bar{\mathcal{K}'}\}$.
  \end{algorithmic}
\end{algorithm}

We have now our main result for the
Dynamic Verified Evaluation of Secret Polynomials.

\begin{restatablebackref}{theorem}{fullVESPoTHM}{thm:full}{app:proofs}
  Under the $d$-BSDH, DLOG, CRHF and DLM security assumptions of
  \cref{sec:secu}, the protocol composed
  of~\cref{alg:setup,alg:read,alg:update,alg:update2,alg:veval,alg:verif}
  (summarized in \cref{proto:full}) is a fully secure verifiable
  polynomial evaluation scheme, as defined in~\cref{def:secVDPE} and
  the complexity bounds of its algorithms are given in
  \cref{tab:complexity}.
\end{restatablebackref}

\begin{table}[!ht]\centering
\renewcommand*{\arraystretch}{1.25}
\caption{Complexity bounds for verifiable dynamic and ciphered
  polynomial evaluation {\footnotesize
    (function of the degree $d$ of the polynomial, for groups of
    supposed constant cardinality: number of group elements/arithmetic
    operations)}.}\label{tab:complexity}
  \adjustbox{max width=\textwidth}{\centering
    \begin{tabular}{ccccc}
      \toprule
      \multicolumn{2}{c}{}& Server & Communication & Client \\
      \midrule
      \multicolumn{2}{c}{Storage} & \bigO{d} & & \bigO{1} \\
      \midrule
      \multirow{5}{*}{\rotatebox[origin=c]{90}{Comput.}}
      &\Setup	& \bigO{d}	& \bigO{d}	& \bigO{d} \\
      &\Read	& \bigO{\log(d)}	& \bigO{\log(d)}	& \bigO{\log(d)} \\
      &\Update	& \bigO{\log(d)}	& \bigO{\log(d)}	& \bigO{\log(d)} \\
      &\deltaUpdate	& \bigO{1}	& \bigO{1}	& \bigO{1} \\
      &\Eval/ \Verif	& \bigO{d}	& \bigO{1}	& \bigO{1} \\
      \bottomrule
    \end{tabular}
  }
\end{table}

For the complexity bounds we still consider the cardinality of the
coefficient domain to be a constant (so that, again, even exponentiations not
involving the degree are considered constant) and we also consider
that one encryption/decryption with the linearly homomorphic
cryptosystem requires a number of arithmetic operations constant with
respect to the degree.

\subsection{Experiments}\label{ssec:full-exper}
To assess the efficiency of our protocol, we
implemented~\cref{proto:full} using the following libraries%
\footnote{
\url{https://gmplib.org},
\url{https://linbox-team.github.io/fflas-ffpack},
\url{https://github.com/scipr-lab/libsnark.git},
\url{https://github.com/relic-toolkit/relic}.
}:
\texttt{gmp-6.2.1} for modular operations,
\texttt{fflas-ffpack-2.4.3} for linear algebra,
\texttt{relic-0.6.0} for Paillier's cryptosystem
and pairings (we used a ``bn-p254'' pairing),
\texttt{libsnark} (commit \texttt{2af4402}) for baseline polynomial
evaluation verification.
Our source code to perform these experiments is available via the
following GitHub repository: \url{https://github.com/jgdumas/vespo}.

\begin{table}[htbp]\centering
\caption{Comparative behaviors of pairings and Paillier system on the
  Server and Client sides with a $254$-bits group size for the
  protocol of~\cref{proto:full}
{\footnotesize (column 'pows' is the time to
  perform the lhs exponentiations (by $s-r$ and by
  $D(\zeta)\alpha[j]+c[j]$);
  column 'c' times the matrix geometric sum; and column 'D' times
  the single Paillier's deciphering;
  below are some baseline comparisons:
  'Horner' is a witness direct evaluation in that group,
  '\texttt{libsnark}' is an unciphered and static polynomial evaluation verification.
  Each experiment was performed $11$ times and we report the median value,
  with a maximum variance lower than $16.4$\% between runs)}.
}\label{tab:lintests}\setlength{\tabcolsep}{2pt}
\adjustbox{max width=\textwidth}{\centering
    \begin{tabular}{crrrrccc}
      \toprule
      \multirow{3}{*}{Degree} &
      \multicolumn{4}{c}{Server Certification} & \multicolumn{3}{c}{Client
        Verification}\\
\cmidrule(lr){2-5}
      & \multicolumn{2}{c}{1 core} &  \multicolumn{2}{c}{4 cores}
      & \multicolumn{3}{c}{1 core} \\
\cmidrule(lr){2-3}\cmidrule(lr){4-5}\cmidrule(lr){6-8}
      & \multicolumn{1}{c}{$\zeta$} & \multicolumn{1}{c}{$\xi$} &
      \multicolumn{1}{c}{$\zeta$} & \multicolumn{1}{c}{$\xi$} & $D$ & $c$ & pows  \\
      \midrule
256 & 0.12s & 0.08s & 0.04s & 0.03s & \multirow{10}{*}{0.9ms} & \multirow{10}{*}{$<$0.1ms} & \multirow{10}{*}{0.7ms}  \\
512 & 0.24s & 0.15s & 0.07s & 0.05s & & & \\
1024 & 0.48s & 0.30s & 0.13s & 0.10s & & & \\
2048 & 0.95s & 0.61s & 0.26s & 0.18s & & & \\
4096 & 1.90s & 1.22s & 0.51s & 0.35s & & & \\
8192 & 3.82s & 2.44s & 1.01s & 0.70s & & & \\
16384 & 7.58s & 4.87s & 2.02s & 1.40s & & & \\
32768 & 15.24s & 9.78s & 4.05s & 2.75s & & & \\
65536 & 30.55s & 19.58s & 8.06s & 5.45s & & & \\
131072 & 60.82s & 39.02s & 16.15s & 10.90s & & & \\
\end{tabular}}
\adjustbox{max width=\textwidth}{\centering
\begin{tabular}{lcrrrrr}
\toprule
	& Client & \multicolumn{4}{c}{Server
          (1 core)} & \multicolumn{1}{c}{Proof}\\
\cmidrule(lr){3-6}
	& 1 core & $d^\circ$\hfill 256 & 1024 & 8192 & 131072 & \multicolumn{1}{c}{size}\\
\midrule
\multicolumn{2}{c}{Horner (no verif., no crypt.)} & $<$0.1ms &0.2ms&
1.6ms& 32.0ms & \multicolumn{1}{c}{-}\\
\texttt{libsnark} (no crypt.) & 3.8ms & 0.04s & 0.12s & 0.74s &10.57s & 287B\\
Here (v. \& c. \& dyn.) & 1.6ms & 0.20s &  0.78s & 6.26s & 99.84s & 320B\\
\bottomrule
\end{tabular}}
\end{table}

To observe the effect of the chosen homomorphic systems (Paillier with an RSA
modulus size of $2048$ bits and the pairing), we ran the experiments,
on a single core of an intel Gold 6126 \@ 2.6GHz
for the Client and Horner computations and on one or four
cores for the Server
(the parallelization of the prefix-like Server part
of~\cref{alg:veval} is given in~\cref{app:parprefix}).
In~\cref{tab:lintests}, we thus compare the Server time to the Client time
of our protocol, to that of a simple (witness) polynomial evaluation
(Horner-like) in this group and of an unciphered static polynomial evaluation
with a SNARK (a ciphered evaluation with these SNARK would require to
arithmetize the Homomorphic cryptosystem and seems still out of
reach).

First of all, of course, the Server time, using homomorphic arithmetic,
can be several orders of magnitude slower than the simple polynomial
evaluation, while indeed being clearly linear.
Second, for the protocol itself, we see that both homomorphic
evaluations of the Server are quite similar, even if the Paillier
cryptosystem is more expensive for large modulus. Then, on the Client
side and for the considered degrees, the dominant computation is that
of a single Paillier's deciphering (and that the only part
in practice potentially non-constant in the degree is by far the most
negligible).
Third our Client is even faster than an unciphered one (we use less
pairings than \texttt{libsnark})
and for a large
enough degree, we can observe the Client time to win over
the linear time pure polynomial evaluation.
Also, our ciphered Server slowdown remains within a factor close to
four (or only two without Paillier) when compared to the static and unciphered
one.

\section{Low Server storage dynamic PoR}\label{sec:por}
Recall that Proofs of Retrievability (PoR)
allow a Client with limited storage, who has outsourced her data to an
untrusted Server, to confirm via an efficient \Audit protocol
that the data is still being stored in its entirety.
The lower bound of \cite[Theorem~4]{Anthoine:hal-02875379} proves that a
tradeoff is inevitable between low/high audit cost and high/low
storage overhead.
The dynamic PoR schemes of
\cite{CashPOR13,Shi:2013:orampor} optimize for fast audits. They
incur a large \bigO{N} storage overhead
on the Server, but can perform
audits with only $(\log N)^{\bigO{1}}$ communication and computation for the
Client and Server.

\begin{table}[!ht]\centering
\caption{Attributes of some selected DPoR schemes} \label{tab:attributes}
\renewcommand*{\arraystretch}{1.25}
\setlength{\tabcolsep}{2pt}
  \adjustbox{max width=\textwidth}{\centering
    \begin{tabular}{lccccc}
      \toprule
      \multirow{3}{*}{Protocol} & \multicolumn{2}{c}{Server} & &
      \multicolumn{2}{c}{Client}\\
\cmidrule(lr){2-3}\cmidrule(lr){5-6}
      & Extra & \Audit & \Audit & \multirow{2}{*}{Storage} & \Audit\\
      & Storage & Comput. & Comm. & & Comput.\\
      \midrule
      \cite{Shi:2013:orampor}&
      $\bigO{N}$ & {\bf\color{darkgreen} $\bigO{\log{N}}$ }& {\bf\color{darkgreen} $\bigO{\log{N}}$ }& {\bf\color{darkgreen} $\bigO{1}$ }& {\bf\color{darkgreen} $\bigO{\log{N}}$ }\\
      \cite{Anthoine:hal-02875379}&
      {\bf\color{darkgreen} $\smallo{N}$ } & $N+\smallo{N}$ & $\bigO{\sqrt{N}}$ & {\bf\color{darkgreen} $\bigO{1}$ } & $\bigO{\sqrt{N}}$ \\
      Here &
      {\bf\color{darkgreen} $\smallo{N}$ } & $N+\smallo{N}$ & {\bf\color{darkgreen} $\bigO{\log{N}}$ } & {\bf\color{darkgreen} $\bigO{1}$ } & {\bf\color{darkgreen} $\bigO{\log{N}}$ }\\
      \bottomrule
    \end{tabular}
}
\end{table}

\begin{table*}[htbp]\centering
  \caption{Private verifiable Client/Server DPoR protocol with low storage Server}\label{protoPor}
\fbox{\adjustbox{max width=\textwidth}{\centering
  \begin{tabular}{cccc}
    & Server & Communications & Client \\
    \midrule
    \multirow{9}{*}{\Init}
&&& {\bf Input}: $p$ prime, cryptosystem $E/D$\\
&&& {\bf Input}: $\MM \in \Z_p^{m \times n}$\\
&&& $\svec \random{\Z_p^*}$, \(\Transpose{\uu}\gets[\svec^i]_{i=0}^{m-1} \in \Z_p^m\) \\

&&& $\Transpose{\vv}\gets\Transpose{\uu}\MM\in\Z_p^n$\\

& \multicolumn{3}{c}{\(\begin{array}{:rcl:}
      \cdashline{1-3}
      & \multirow{2}{*}{\ovalbox{\deltaSetup}} & \plleftarrow[\commlength]{} \convertintopoly(\vv)\footnotemark[4]\hspace*{1cm}\\
 \hspace*{2cm}\servstate  \plleftarrow[\commlength]{} & & \plrightarrow[\commlength]{} \clstate \\
      \cdashline{1-3}
    \end{array}\)} \\

&$T_\MM \gets\mtcreate(\MM)$ &    \pleftarrow{ \MM}& $r_M \gets \mtrootfromleaves(\MM)$\\
& {\bf Output: } $\servstate, \MM, T_\MM  $&&  {\bf Output: } $\clstate, \svec,r_M $ \\

    \midrule
    \multirow{5}{*}{\Write}
    &$L_{M_{ik}}\gets\mtuncles(k+i{\cdot}n,M,T_M)$ &
    \pleftarrow{i,k} & {\bf  Input:} $i,k,\MM'_{ik}$ \\
    && \prightarrow{\MM_{ik},L_{\MM_{ik}},\ww_{k}}  &$r_M \checks{=} \mtrootfrompath(k+i{\cdot}n,M_{ik},L_{M_{ik}})$\\
   &&& $\delta\gets \gamma^i(\MM'_{ik}-\MM_{ik})$\\
& \multicolumn{3}{c}{\(\begin{array}{:rcl:}
      \cdashline{1-3}
 \servstate  \plrightarrow[\commlength]{}   & \multirow{2}{*}{\ovalbox{\deltaUpdate}} &\plleftarrow[\commlength]{} \delta,  \clstate \\
\hspace*{2cm} \servstate\plleftarrow[\commlength]{} & & \plrightarrow[\commlength]{} \clstate \hspace*{2cm} \\
      \cdashline{1-3}
 \end{array}\)} \\

    & $\MM_{ik}\gets\MM'_{ik}$, $T_\MM\gets\mtupdateleaf(k+i{\cdot}n,\MM'_{ik},T_\MM)$&\pleftarrow{\MM'_{ik}} & $r_M \gets \mtrootfrompath(k+i{\cdot}n,M'_{ik} ,L_{M_{ik}})$\\
& {\bf Output: } $\servstate, \MM, T_\MM  $&&  {\bf Output: } $\clstate,r_M $ or \reject \\
    \midrule
    \multirow{4}{*}{\Audit} & form
    \(\xx\gets[r^k]_{k=0}^{n-1}{}^\intercal\in\Z_p^{n}\), then $\yy\gets\MM\xx$
    & \pleftarrow{r}
& $r\random{\Z_p^*}$ s.t. $(r\matr{\Phi}-I_2)\in{GL_2(\Z_p)}$ \\
&{\(\begin{array}{:rc:}
      \cdashline{1-2}
  \servstate, r \longrightarrow    & \multirow{2}{*}{\ovalbox{\Eval}}  \\
\zeta,\bar{\xi} \longleftarrow &   \\
      \cdashline{1-2}
 \end{array}\)} & \prightarrow{\yy,\zeta,\bar{\xi}}
& %
{\(\begin{array}{:cl:}
      \cdashline{1-2}
  \multirow{2}{*}{\ovalbox{\Verif}} & \longleftarrow \clstate, r,\zeta,\bar{\xi} \\
 & \longrightarrow  D_{\sk}(\zeta) \, or \, \reject \\
      \cdashline{1-2}
 \end{array}\)} \\

&&&  $\Transpose{\uu}y\checks{=} D_{\sk}(\zeta)$\\
&&&  {\bf Output: } \accept or \reject\\

  \end{tabular}
}}
\end{table*}

Instead, \cite{Anthoine:hal-02875379}
optimizes for small storage; their scheme
has only sub-linear storage overhead of \bigO{N/\log N}, but a
higher audit cost of \bigO{N} on the Server, and \bigOsqrt{N} Client
time and communication. The authors demonstrate that, for
reasonable deployment scenarios on commercial cloud platforms, the
higher audit cost is more than offset by the greatly reduced costs of
extra persistent storage, especially if audits are only performed a few
times per day.

We here further improve on the low storage overhead approach of
\cite{Anthoine:hal-02875379}, by our scheme with a small
\smallo{N} storage overhead, but only \bigO{\log N}
communication and Client computation cost for audits.
That is, our new protocol still benefits from small storage
overhead, while effectively pushing the higher computational cost of audits
(which is inevitable from the lower bound) entirely off the Client and
onto the Server.
These savings are highlighted in \cref{tab:attributes}.

An easy argument demonstrates that our \bigO{\log N} Client
cost for audits is optimal. If each audit has \smallo{\log N}
cost (and thus transcript size), then the total number
of possible transcripts is \smallo{N}, which is a contradiction with the
definition of retrievability; not every
$N$-bit database could be recoverable via independent audit transcripts.

\subsection{Matrix based approach for audits} \label{ssec:First-PoR}
Here we summarize the DPoR of~\cite{Anthoine:hal-02875379} upon which
our new scheme is based.
The premise is to
treat the data, consisting of $N$ bits organized in machine words, as
a matrix $\MM\in\Z_p^{m\times n}$, where $\Z_p$ is a suitable finite
field of size~$p$.
Crucially, the choice of ring $\Z_p$ does not require any modification
to the raw data itself; that is, any element of the matrix $\MM$ can
be retrieved in $O(1)$ time from the underlying raw data storage.
The scheme is based on the commutativity of matrix-vector products. During
an $\Init$ phase, the Client chooses a secret vector $\uu$ of
dimension $m$ and
computes $\Transpose{\vv}=\Transpose{\uu}\MM$; both vectors $\uu$ and $\vv$ are then
stored by the Client for later use, while the Server stores the original
data and hence the matrix $\MM$ in the clear.
Reading or updating individual entries in $\MM$ (\Read, and
\Write~protocols in the DPoR case), can be performed
efficiently with the use of Merkle hash trees and from the observation
that changing one element of $\MM$ only requires changing one entry in
the Client's secret control vector~$\vv$.
To perform an \Audit, the Client and Server engage in a 1-round protocol:
\begin{enumerate}[leftmargin=3\labelsep]
  \item Client chooses a random vector $\xx$ of dimension $n$, and sends
  $x$ to Server.
  \item Server computes $\yy = \MM\xx$ and sends the dimension-$m$
  vector $y$ back to Client.
  \item Client computes two dot products $\Transpose{\uu}\yy$ and $\Transpose{\vv}\xx$,
  and checks that they are equal.
\end{enumerate}
The proof of retrievability relies on the fact that observing several
successful audits allows, with high probability, recovery of the correct
matrix $\MM$, and therefore of the entire database.
The communication costs are \bigO{n} and \bigO{m} in steps 1 and 2
respectively, and the Client computation in step 3 is \bigO{m+n},
resulting in \bigOsqrt{N} total communication and Client computation
when optimizing the matrix dimensions to roughly $m=n=\sqrt{N}$.

\begin{table*}[htbp]\centering
  \caption{Modification of the DPoR audit protocol, with 254-bits groups,
    2048-bits Paillier, on a Gold 6126 \@ 2.6GHz \& 10 GB/core
    ({\footnotesize
  real time are median values for a single run;
  each experiment was performed $11$ times;
  the maximum relative difference between the runs was
  at most 3.6\%}).}\label{table:results}
  \renewcommand{\arraystretch}{1} %
  \setlength{\tabcolsep}{4pt}
  \adjustbox{max width=\textwidth}{\centering
    \begin{tabular}{lcccc}
      \toprule
      \multicolumn{1}{c}{Database} & $1$GB     & $10$GB    &
      $100$GB   & $1$TB \\
      \toprule
      \multicolumn{5}{c}{Private-verified audit using $57$-bits
        prime~\cite[Figure~1 \& Tables~5-6-7]{Anthoine:hal-02875379}\footnotemark[5]} \\
      \midrule
      {Matrix view}
      & $12339{\times}12432$ & $39131{\times}39200$ & $123831{\times}123872$ & $396281{\times}396368$ \\
      {Server extra storage}
      & \color{darkgreen}  ${<}0.01$\% & \color{darkgreen}  ${<}0.01$\%  & \color{darkgreen}  ${<}0.01$\%  &\color{darkgreen}   ${<}0.01$\%   \\
      {Client Storage}
      & $169$KB & $535$KB & $1\,693$KB & $5\,418$KB \\
      {Server Audit (1/12 cores)}
      &\color{darkgreen}  $0.29$s/$0.04$s &\color{darkgreen}  $2.68$s/$0.30$s &\color{darkgreen}  $29.04$s/$3.36$s&\color{darkgreen}  $219.7$s/$41.48$s \\
      {Communications}
      &\color{darkgreen}  $169$KB & $535$KB & $1\,693$KB & $5\,418$KB \\
      {Client Audit (1 core)}
      &\color{darkgreen}  $0.6$ms &\color{darkgreen}  $1.7$ms &\color{darkgreen}  $5.3$ms & $18.3$ms \\
      \toprule \multicolumn{5}{c}{Square Dynamic-ciphered delegated polynomial
        evaluation with $254$-bits groups of~\cref{protoPor}\footnotemark[6]} \\
      \midrule
      {Matrix view}
      & $5815{\times}5816$ & $18390{\times}18390$ & $58154{\times}58154$ & $186092{\times}186093$\\
      {Server extra storage}
      & $0.12$\% & $0.04$\% & $0.01$\% & ${<}0.01$\% \\
      {Client storage}
      & \color{darkgreen} $0.94$KB  & \color{darkgreen} $0.94$KB  & \color{darkgreen} $0.94$KB  & \color{darkgreen}  $0.94$KB  \\
      \multicolumn{1}{l}{Server Audit (1/12 cores): matrix-vector step}
      & $1.1$s/$0.2$s & $11.3$s/$1.3$s & $113.4$s/$12.9$s & $1\,152.5$s/$131.1$s \\
      \multicolumn{1}{l}{Server Audit (1/12 cores): polynomial step}
      &  $4.4$s/$0.5$s  &  $13.5$s/$1.4$s  &  $42.6$s/$4.2$s  &  $141.7$s/$13.4$s  \\
     {Communications}
      & $181$KB & $571$KB & $1\,803$KB & $5\,770$KB \\
      \multicolumn{1}{l}{Client Audit (1 core): dotproduct step}
      & $3.2$ms & $8.4$ms & $13.1$ms & $37.9$ms  \\
      \multicolumn{1}{l}{Client Audit (1 core): polynomial step}
      &  $1.7$ms  &  $1.7$ms  &  $1.7$ms  &  $1.7$ms  \\
 \toprule
 \multicolumn{5}{c}{Rectangular Dynamic-ciphered delegated polynomial
        evaluation with $254$-bits groups of~\cref{protoPor}\footnotemark[6]} \\
      \midrule
      {Matrix view}
      & $6599{\times}5125$ & $7265{\times}46551$ & $7929{\times}426519$ & $8600{\times}4026778$ \\
      {Server extra storage}
      &  $0.11$\%  &  $0.10$\%  &  $0.09$\%  &  $0.08$\%  \\
      {Client storage}
      & \color{darkgreen}  $0.94$KB  & \color{darkgreen}  $0.94$KB  & \color{darkgreen}  $0.94$KB  & \color{darkgreen}  $0.94$KB  \\
      \multicolumn{1}{l}{Server Audit (1/12 cores): matrix-vector step}
      & $1.1$s/$0.2$s & $11.3$s/$1.3$s  & $113.2$s/$12.8$s  & $1\,147.9$s/$130.7$s \\
      \multicolumn{1}{l}{Server Audit (1/12 cores): polynomial step}
      &  $3.8$s/$0.4$s   &  $35.5$s/$3.6$s  &  $324.1$s/$30.6$s  &  $3\,064.8$s/$283.6$s  \\
      {Communications}
      & $205$KB &\color{darkgreen}  $226$KB &\color{darkgreen}  $246$KB &\color{darkgreen}  $267$KB \\
      \multicolumn{1}{l}{Client Audit (1 core): dotproduct step}
      & $3.7$ms & $4.0$ms  & $4.4$ms  & \color{darkgreen} $4.8$ms \\
      \multicolumn{1}{l}{Client Audit (1 core): polynomial step}
      &  $1.7$ms  &  $1.7$ms  &  $1.7$ms  & \color{darkgreen}  $1.7$ms  \\
      \bottomrule
    \end{tabular}
  }
\end{table*}

While this square-matrix setup is the basic protocol presented by
\cite{Anthoine:hal-02875379},
the authors also discuss a potential improvement in communication
complexity. Instead of $\xx$ being uniformly random over $\Z_p^n$, it
can instead be a \emph{structured} vector formed from a single random
element $r\in\Z_p$ as $\xx=[r^i]_{i=1}^{n}$. Then the communication on
step 1 is reduced to constant, and hence the total communication depends
only on the row dimension \bigO{m}. By choosing a rectangular matrix
$\MM$ with few rows and many columns, the communication can be made
arbitrarily small.
The tradeoff for this reduction in communication complexity is higher
Client storage of the control vector $\vv$ as well as higher Client
computation cost for the $n$-dimensional dot product $\Transpose{\vv}\xx$.
In \cite{Anthoine:hal-02875379}, the authors found that the savings in
communication were not worth the higher Client storage and computation,
and their experimental evaluation was based on the square matrix version
with overhead \bigOsqrt{N}.

\footnotetext[4]{Converts the vector $\vv$ into the polynomial
  $P(x)=\sum_{i=0}^{n-1} \vv_ix^i$.}
\subsection{Bootstrapping Client via VESPo}

Now we show how to modify the reduced communication version of the DPoR
protocol of \cite{Anthoine:hal-02875379} just presented in order to
eliminate the costly Client storage of $\vv\in\Z_p^n$ and computation of
$\Transpose{\vv}\xx$ during audits.
Our improved protocol is based on the observation that, when the audit
challenge vector $\xx$ is structured as $\xx=[r^i]$, then the expensive
Client dot product computation of $\Transpose{\vv}\xx$ is actually a polynomial
evaluation: if the entire of $\vv$ are the coefficients of a polynomial
$P$, then $\Transpose{\vv}\xx$ is simply $P(r)$.
We therefore eliminate the \bigO{n} Client persistent storage and
computation cost during audits by outsourcing the (encrypted) storage of
vector $\vv$ and computation of $\Transpose{\vv}\xx=P(r)$ with our novel protocol
for dynamic, encrypted, verifiable polynomial evaluation scheme of
\cref{proto:full}.
The obtained private-verification DPoR protocol,
combining that of \cite{Anthoine:hal-02875379} with our ciphered
polynomial evaluation in \cref{sec:full}, is
presented in~\cref{protoPor}.

\begin{restatablebackref}{theorem}{DPorTHM}{thm:porpol}{app:proofs}
  The protocol of~\cref{protoPor} is correct and sound under the
  $d$-BSDH, DLOG, CRHF and DLM security assumptions.
\end{restatablebackref}

\subsection{Experiments}\label{ssec:por-exper}
We now compare our modification of the DPoR protocol with the one
in~\cite{Anthoine:hal-02875379}.
\cref{table:results} has three blocks of experiments, each for four
database sizes ranging from $1$GB to $1$TB.
The first block of
experiments is a run of the original statistically secure DPoR protocol
with two dotproducts for the verification, considering the matrix as
$56$ bits elements modulo a $57$-bits prime.
The second block of experiments is our new modification, but still
using close to square matrices. Subject now to computational security,
we have to use a larger coefficient domain, namely here a $254$-bits
prime (with associated bilinear groups and a $2048$-bits Paillier
modulus, both estimated equivalent to a $112$-bit computational
security).
We separate the timings of the \Write~phase in two phases, the
remaining linear algebra phase and the new polynomial evaluation
phase (\deltaUpdate).
In the third block of experiments we use a more rectangular matrix,
trying to reduce communications while not increasing too much the
Server computational effort.

Overall, we see first in \cref{table:results}, that changing the
coefficient domain size increases the computational effort of the
Server in the linear algebra phase. Still, reducing the dimension of
the dotproduct for the Client, as shown in he third block, allows the
Client to be faster for databases larger than $100$GB.
In any case, the Client audit computational effort is never larger than a few
milliseconds and thus the dominant part is most certainly
communications.
On this aspect, we see that our modification allows for large
reductions in both the Client storage (even with square matrices) and
the overall communications.
Indeed, the Client private state is the vector dimension, the
Paillier's private key, twelve group elements
and two Merkle tree roots;
while the communications are mostly one vector of modular integers in
the smallest dimension.

The price to pay is from about a factor of four (large database) to an
order of magnitude (tiny database) for the Server computations (more
limited losses in the more realistic case where the Server can use
multiple cores).
In any case, the persistent Client storage is going from dozens of MB to less
than one KB, and the communication volume can be decreased by
more than two orders of magnitude.%

\section{Conclusion}

We have presented a protocol verifying publicly a dynamic unciphered
polynomial evaluation and then a protocol verifying privately a
dynamic ciphered polynomial evaluation.
Now, combining efficient and proven dynamicity for ciphered polynomial
with public verifiability raises security issues and reminds an open
problem.
Still, we have also presented a protocol verifying the outsourced
evaluation of secret polynomials. Client verification is of the order
of a few milliseconds and is faster than direct polynomial evaluation
over a small finite field, as soon as the degree of the polynomial is
larger than a few thousand.

This enables us in turn to reduce by several orders of magnitude the
communications, Client storage and Client computations for
state-of-the-art low Server-storage dynamic proofs of retrievability.

\footnotetext[5]{\url{https://github.com/dsroche/la-por}}
\footnotetext[6]{\url{https://github.com/jgdumas/vespo}}
\section*{Acknowledgments}
  We thank Gaspard Anthoine for providing us with some preliminary
  comparisons with the PBC and libpaillier libraries
  and Anthony Martinez for the libsnark baseline benchmarking.
  We thank Jean-Louis Roch for fruitful exchanges about the
  parallelization of the Server side
  and for pointing out \cite{Snir:1986:prefix}.
  Finally we thank the anonymous referees who greatly helped improve
  the paper.
  {This material is based on work supported in part by the
    Agence Nationale pour la Recherche under Grants
    ANR-21-CE39-0006 Sangria
    and ANR-15-IDEX-0002%
    .}
\bibliographystyle{abbrvurl}
\bibliography{pcbibshort}

\appendix
\section{Overview of VESPo exchanges}\label{app:full}

We here recall in~\cref{proto:full}, the summary of the
exchanges
of~\cref{alg:setup,alg:read,alg:update,alg:update2,alg:veval,alg:verif}.
This gives an
overview of our verifiable \& dynamic evaluation of ciphered
polynomials.

With this summary we can refine in~\cref{tab:protocounts} the results
of~\cref{tab:complexity} if we are using Paillier for the LHE
(a Paillier encryption is $1$ modular exponentiation and $3$ modular
multiplications, a Paillier decryption is $1$ exponentiation
and $1$ multiplication, homomorphic multiplication is an
exponentiation and homomorphic addition is a multiplication; then we
approximate~\cref{alg:fiduccia} with $6$ exponentiationsand $40$ modular operations and we
approximate the application of the pairing bilinear map with $1$
exponentiation).

\begin{table}[htbp]
\caption{Dominant terms in operations counts for~\cref{proto:full}
  using Paillier \footnotesize{(a value of $x$ approximates in fact
    $x+o(x)$; then ``Hash'' counts calls to the cryptographic hash
    function, ``mexp'' is for modular exponentiations, ``group'' is
    for the other arithmetic operations)}.}\label{tab:protocounts}
\adjustbox{max width=\textwidth}{\centering
\begin{tabular}{ccccccc}
\toprule
\multirow{2}{*}{Alg.} & \multicolumn{3}{c}{Server} & \multicolumn{3}{c}{Client} \\
\cmidrule(lr){2-4}\cmidrule(lr){5-7}
 & group & mexp & Hash & group & mexp & Hash\\
\midrule
\ref{alg:setup}	& $0$ & $0$ & $2d$ & $17d$ & $6d$ & $2d$	\\
\ref{alg:read}	& $0$ & $0$ & $0$  & $1$   & $1$  & $\lceil\log_2(d)\rceil$	\\
\ref{alg:update}
& $0$ & $0$ & $\lceil\log_2(d)\rceil$
& $18$ & $16$ & $2\lceil\log_2(d)\rceil$ \\
\ref{alg:update2}	& $3$ & $0$ & $0$ & $8$ & $8$ & $0$ \\
\ref{alg:veval}/\ref{alg:verif}	& $3d$ & $4d$ & $0$ & $52$ & $11$ & $0$\\
\bottomrule
\end{tabular}}
\end{table}

\begin{table*}[htbp]\centering\renewcommand{\arraystretch}{1.25}
  \caption{Private \& Dynamic, Ciphered polynomial evaluation, summarizing~\cref{alg:setup,alg:read,alg:update,alg:update2,alg:veval,alg:verif}.}\label{proto:full}
\fbox{\adjustbox{max width=\textwidth}{\centering
  \begin{tabular}{cccc}
    & Server & Communications & Client \\
    \midrule
    \multirow{6}{*}{{\Setup}}&&$\GG_1,\GG_2,\GG_T$ groups of order $p$& $P\in\Z_p[X]$, $1\leq{}d^{\circ}(P)\leq{d}$ \\
    &&pairing $e$ to $\GG_T$,&
    $s\random\Z_p{\setminus}\{0,1\}$, $\vect{\alpha},\vect{\beta}\random\Z_p^2$, $\matr{\Phi}\random\Z_p^{2{\times}2}$,\\
    &&gen. $g_1,g_2,g_T=\gen$&s.t. $(s\matr{\Phi}-I_2)\in{GL_2(\Z_p)}$\\
    &&& Let $\bar{P}(X) \gets \sum_{i=0}^dX^i({p_i}\vect{\alpha}+\matr{\Phi}^i\vect{\beta})$
    \\
 Alg.~\ref{alg:setup}   &&& $W \gets E_{\pk}(P)$, $S \gets [g_1^{s^k}]_{k=0}^{d-1}{\in}\GG_1^d$\\
    &&&$\bar{\mathcal{K}}\gets{g_T^{\bar{P}(s)}}{\in}\GG_T^2$,
    $\bar{H} \gets [g_2^{\bar{p}_i}]_{i=1}^{d}{\in}\GG_2^{2{\times}d}$\\
    &$T_W\gets\mtcreate(W)$
    & \pleftarrow{W,\bar{H},S} &
    $d_p\gets{d\mod{\varphi(p)}}$, $r_W \gets \mtrootfromleaves (W)$\\

&{\bf{Output}} : $\servstate{}  = \{\pk,\GG_{1,2,T},e,W,T_W,S,\bar{H}\}$&& $\clstate{} =\{\pk,\sk,\GG_{1,2,T},g_{1,2,T},e,s,\vect{\alpha},\vect{\beta},\matr{\Phi},\bar{\mathcal{K}},r_W,d_p\}$\\

    \midrule
    \multirow{2}{*}{{\Read}} && \pleftarrow{i} &\\
 Alg.~\ref{alg:read}   &$L_i\gets\mtuncles(i,W,T_W)$ & \prightarrow{w_i,L_i} & $r_W \checks{=} \mtrootfrompath(i,w_i,L_i)$\\
&&& {\bf{Output}} : $p_i \gets D_{\sk}(w_i)$ or {\bf reject}\\
    \midrule
    \multirow{2}{*}{{\Update}} &
& \pleftarrow{i,w'_i,if \,(i>0) \,\bar{H}'_i} & $w'_i \gets E_{\pk}(p'_i)$, $\bar{H}'_i\gets g_2^{p'_i \vect{\alpha}+\matr{\Phi}^i\vect{\beta}}$\\
  Alg.~\ref{alg:update}  &$L_i\gets\mtuncles(i,W,T_W)$ & \prightarrow{w_i,L_i} & $r_W \checks{=} \mtrootfrompath(i,w_i,L_i)$\\
    & $T_W\gets\mtupdateleaf(i,w'_i,T_W)$ && $r_W \gets \mtrootfrompath(i,w'_i,L_i)$\\
    & $w_i \gets w_i'$
      &&$\Delta \gets g_2^{(p'_i-p_i)\alpha}$, $\bar{\mathcal{K}}[j]\gets e(g_1, \Delta[j]^{s^i})\cdot\bar{\mathcal{K}}[j]$\\
&{\bf{Output}} : $\servstate{}  = \{\pk,\GG_{1,2,T},e,W,T_W,S,\bar{H}\}$&& $\clstate{} =\{\pk,\sk,\GG_{1,2,T},g_{1,2,T},e,s,\vect{\alpha},\vect{\beta},\matr{\Phi},\bar{\mathcal{K}},r_W,d_p\}$\\
&&& or {\bf{reject}}\\
    \midrule
    \multirow{1}{*}{{\deltaUpdate}} &
    If $i>0, \bar{H}'_i[j]\gets\Delta[j]\cdot\bar{H}_i[j]$ for $j=1..2$ & \pleftarrow{i,e_\delta,\Delta} &
    $e_\delta \gets E_{\pk}(\delta)$, $\Delta \gets g_2^{\delta\vect{\alpha}}$\\
 Alg.~\ref{alg:update2}   & $w_i\gets{w_i}\cdot{e_\delta}$ &&
    $\bar{\mathcal{K}}[j]\gets e(g_1;\Delta[j]^{s^i})\cdot\bar{\mathcal{K}}[j]$\\
&{\bf{Output}} : $\servstate{}  = \{\pk,\GG_{1,2,T},e,W,S,\bar{H}\}$&& $\clstate{} =\{\pk,\sk,\GG_{1,2,T},g_{1,2,T},e,s,\vect{\alpha},\vect{\beta},\matr{\Phi},\bar{\mathcal{K}},d_p\}$\\
    \midrule
    \multirow{1}{*}{{\Eval/\Verif}}& Form
    $x\gets\Transpose{[1,r,r^2,\ldots,r^d]}$ &
    \pleftarrow{r} & For $r\in\Z_p$ s.t. $(r\matr{\Phi}-I_2)\in{GL_2(\Z_p)}$\\
 \multirow{2}{*}{Alg.~\ref{alg:veval}/\ref{alg:verif}}    &$\zeta\gets\Transpose{W}\boxdot{x}$&&$\vect{c}\gets((r\matr{\Phi})^{d_p+1}-I_2)(r\matr{\Phi}-I_2)^{-1}\vect{\beta}$\\
  &$\bar{\xi}[j]\gets\prod_{i=1}^d \prod_{k=0}^{i-1}
    e(S_{i-k-1};\bar{H}_i[j])^{x_k}$ for $j=1..2$&
    \prightarrow{\zeta,\bar{\xi}} & $\bar{\xi}[j]^{s-r}g_T^{{D_{\sk}(\zeta)}\vect{\alpha}[j]+\vect{c}[j]}\checks = \bar{\mathcal{K}}[j]$ for $j=1..2$\\

&&& {\bf Output} : $D_{\sk}(\zeta)$ or {\bf reject}\\
  \end{tabular}
}}
\end{table*}

\section{Proofs of the propositions and theorems}\label{app:proofs}
Now, we give the proofs of the propositions
in~\cref{sec:cipher,sec:dynamic} and of our main theorems for our
private and dynamic ciphered polynomials evaluation protocol
and our low Server storage and audit complexity DPoR.

\cKZGverifiable*
\begin{proof}
{\bf Correctness}. First, $\zeta=\Transpose{W}\boxdot{x}=\prod_{i=0}^d
E(p_i)^{(r^i)}=E(P(r))$.
Then, second,
$\xi=\Transpose{H}\odot{x}=\prod_{k=0}^{d-1} g_2^{T_{k,P}(s)r^k}=g_2^{Q_P(s,r)}$,
by~\cref{prop:bivariate}.
Therefore, the verification is
that $g_T^{Q_P(s,r)(s-r)+P(r)}\checks{=}g_T^{P(s)}$ and this is guaranteed
by~\Cref{eq:KZG}.

{\bf Soundness}.
Let
$\left\langle{}g_2,g_2^s,g_2^{s^2},\ldots,g_2^{s^t}\right\rangle\in\GG_2^{t+1}$
be a t-BSDH instance and suppose that there exists an attack to the
\Audit protocol.

Let  $[p_0,\ldots,p_t]\random\Z_p^{t+1}$ for a degree $t$
polynomial and $d=t$.
Then compute directly $W=E(P)$, $T_{k,P}=\sum_{i=k+1}^t p_i
Y^{i-k-1}=\sum_{j=0}^{t-1-k}t_{k,j}Y^j$
and homomorphically compute:
\[\mathcal{K}=e\left(g_1;\left\langle{}g_2,g_2^s,g_2^{s^2},\ldots,g_2^{s^t}\right\rangle\odot[p_0,\ldots,p_t]\right),\]
together with $H=[h_k]$, where
$h_k=\left\langle{}g_2,g_2^s,g_2^{s^2},\ldots,g_2^{s^{t-1-k}}\right\rangle\odot[t_{k,0},\ldots,t_{k,t-1-k}]$.
These inputs are indistinguishable from a generic setup of the protocol
of~\cref{proto:cKZG} and can thus be given to its attacker.

Finally, select a random evaluation point $r$ and compute
$(\zeta,\xi)$.
The supposition is that an attacker of the \Audit{} part of the
protocol can get $(\zeta',\xi')$, with some advantage, such that
$(D(\zeta'),\xi')\neq(D(\zeta),\xi)$, even though both would be
passing the verification.
Now, on the one hand, if $D(\zeta')=D(\zeta)$, then $\xi\neq{\xi'}$
and it must be that
$e(g_1^{s-r};\xi)g_T^{D(\zeta)}=\mathcal{K}$ and
$e(g_1^{s-r};\xi')g_T^{D(\zeta)}=\mathcal{K}$.
Therefore, if $r\neq{s}$, then
$e(g_1^{s-r};\xi)=e(g_1^{s-r};\xi')$ contradicts the fact that
$\xi\neq\xi'$; so $r=s$, and the secret
can be exposed.
On the other hand, if $D(\zeta')\neq{D(\zeta)}$, then it means that we
must have the equality
$(e(g_1;\xi)/(e(g_1;\xi'))^{s-r}=g_T^{D(\zeta')-D(\zeta)}$ and therefore:
\(
\left(\frac{e(g_1;\xi)}{e(g_1;\xi')}\right)^{\frac{1}{D(\zeta')-D(\zeta)}}=g_T^{\frac{1}{s-r}}
\).
This proves that the adversary would solve the t-BSDH
$\left\langle{}{-}r,\gen^{\frac{1}{s-r}}\right\rangle$ challenge with
the same advantage.
\myqed\end{proof}

From this proof, one can see that using a decipherable partially
homomorphic function for the coefficients of $P$ is required for the
soundness (otherwise one could not compute the exponentiation on
$\xi/\xi'$).

\csClear*
\begin{proof}
\textbf{Correctness.}
First, \eqref{eq:merkle} gives the correctness of \Read.
For \Update, \eqref{eq:mtupl} provides the correctness of the hash tree.
Then, with $\delta=p'_i-p_i$, the new polynomial is
$P'(s)=P(s)+\delta{s^i}$, so that the key is updated as
${\mathcal{K}}_1'={\mathcal{K}}_1\cdot{e(g^{\delta{s^i}};g)}$.
Now for the evaluation, first,
$\xi=\prod_{i=1}^d\prod_{k=0}^{i-1}S_{i-k-1}^{p_ix_k}=g^{\sum\sum{s^{i-k-1}p_ix_k}}=g^{Q_P(r,s)}$
and, second, we have that:
\begin{multline*}e(\xi;{\mathcal{K}}_2/g^r)\gensym^{\zeta}=e(\xi;g^{s-r})\gensym^{P(r)}
=\\\gensym^{Q_P(r,s)(s-r)+P(r)}=\gensym^{P(s)}.\end{multline*}
Hence we see that
$e(\xi;{\mathcal{K}}_2/g^r)\gensym^{\zeta}={\mathcal{K}}_1$ and,
therefore, the protocol is correct.

\textbf{Soundness.}
First for the \Read/\Update~parts. Suppose an attacker can provide
$p'_i\neq{p_i}$ that passes the Merkle root check.
This would violate the soundness property of \cref{eq:mtsound},
which is derived from the collision resistance of the underlying hash
function.

Second, for the \Eval/\Verif~parts.
Let
$\left\langle{}g,g^s,g^{s^2},\ldots,g^{s^t}\right\rangle\in\GG^{t+1}$
be a t-BSDH instance.
For the setup phase, just set $d=t$ and then randomly select
$[p_0,\ldots,p_t]\random\Z_p^{t+1}$.
Then set
$S=\left\langle{}\GG,g,g^s,g^{s^2},\ldots,g^{s^t}\right\rangle$
and
\[\mathcal{K}_1=e\left(\left\langle{}g,g^s,g^{s^2},\ldots,g^{s^t}\right\rangle\odot[p_0,\ldots,p_t];g\right).\]
These inputs are indistinguishable from generic inputs to the protocol
of~\cref{protoDynClear}.
For any number of update phase, randomly select $p'_i$ (or $\delta$), receive
$p_i$ and $L_i$ from the Server, compute
${\mathcal{K}}_1'={\mathcal{K}}_1e(S_i^{\delta};g)$ and refresh $r_p$.
Finally, select a random evaluation point $r$, compute $(\zeta,\xi)$ and
call an attacker of the \Eval~part of the protocol to get
$(\zeta',\xi')$ such that $(\zeta',\xi')\neq(\zeta,\xi)$, even
though both are passing the verification.
If $\zeta'=\zeta$, then as $\xi\neq\xi'$ it must be that $r=s$ and the
secret is revealed;
otherwise, $\zeta'\neq{\zeta}$
and we have both $e(\xi';{\mathcal{K}}_2/g^r)e(g;g )^{\zeta'}=
{\mathcal{K}}_1$, on the one hand,
and ${\mathcal{K}}_1=e(\xi;{\mathcal{K}}_2/g^r)e(g;g )^{\zeta}$, on
the other hand.
This gives  $e(\frac{\xi'}{\xi};g^{s-r})=e(g^{\zeta-\zeta'};g)$ and thus
$e\left((\frac{\xi'}{\xi})^{s-r};g\right)=e(g^{\zeta-\zeta'};g)$.
Finally, we have that:
\(
e\left(\frac{\xi}{\xi'};g\right)^{\frac{1}{\zeta'-\zeta}}=\gensym^{\frac{1}{s-r}}
\).
This proves that the adversary would solve the t-BSDH
$\left\langle{}{-}r,\gensym^{\frac{1}{s-r}}\right\rangle$ challenge with
the same advantage.
\myqed\end{proof}

\compClear*
\begin{proof}
The setup phase requires the Client to perform one polynomial
evaluation and $d$ exponentiations for $O(d)$ arithmetic operations,
together with the computation of the Merkle tree on both sides, for
$O(d)$ hashing operations.

  For the update phase, the Client computes the root of the
  Merkle tree from the new value $p_i+\delta$ and the path $L_i$ given
  by the Server in \bigO{\log(d)}.
  She also has to compute an exponentiation and a product in
  $\Z_p[X]$, this is in~\bigO{1}.

For the verification phase, communications are just $3$ group
elements. The Client work is only $2$ pairing and $2$
exponentiations and $1$ product.

Now for the Server.
First, computing $\zeta$ is $d+1$  homomorphic
multiplications and $d$  additions.
Second, the Server has to compute $\xi=\prod_{i=1}^d \prod_{k=0}^{i-1}
S_{i-k-1}^{p_ix_k}=\prod_{i=1}^d  \left (\prod_{k=0}^{i-1}
  S_{i-k-1}^{r^k} \right )^{p_i}$.
Therefore,
one can use a Horner-like prefix
computation~\cite{Kahan:1999:divdiff}: consider $t_0=1$, and
$t_i=S_{i-1}\cdot{t_{i-1}^r}$, then $t_1=S_0$, $t_2=S_1S_0^r$ and therefore
$t_i=S_{i-1}(S_{i-2}\ldots(S_2(S_1S_0^r)^r)^r\ldots)^r=\prod_{k=0}^{i-1}S_{i-k-1}^{r^k}$.
Thus one can use the following~\cref{alg:prefixdiff} to compute $\xi$.

\begin{algorithm}[!ht]\caption{Homomorphic linear prefix evaluation of the
    difference polynomial}\label{alg:prefixdiff}
\begin{algorithmic}[1]
\REQUIRE $r$, $[S_0,\ldots,S_{d-1}]$, $[p_1,\ldots,p_d]$.
\ENSURE $\xi=\prod_{i=1}^d  \left (\prod_{k=0}^{i-1}  S_{i-k-1}^{r^k} \right )^{p_i}$.
\STATE $\xi=1$; $t=1$;
\FOR{$i=1$ \TO $d$}
\STATE $t\leftarrow{S_{i-1}\cdot{t^r}}$;
\hfill\COMMENT{$t_i=\prod_{k=0}^{i-1}S_{i-k-1}^{r^k}$}
\STATE $\xi\leftarrow\xi\cdot t^{p_i}$.
\ENDFOR
\RETURN $\xi$.
\end{algorithmic}
\end{algorithm}

Computing $\xi $ then requires at most $2d$ exponentiations and $2d$
multiplications.\myqed
\end{proof}

\fullVESPoTHM*
\begin{proof}
First of all, we have that:
\[\left \{\begin{array}{l} H'_i=H_i. \Delta\\w'_i=w_i.e_{\delta}  \end{array} \right .
\Leftrightarrow \left \{\begin{array}{l} \Delta= H'_i.H_i^{-1}\\e_{\delta}=w'_i.w_i^{-1} \end{array}\right.\]
Therefore, it is equivalent to consider the protocols using only
\cref{alg:update} or only \cref{alg:update2} or any combinations of
both. Also, the \Read part is identical to that
of~\cref{protoDynClear} and so are the associated security proofs.

{\bf Correctness.}
For the \Update~operation, %
$\bar{P}'(s)=\bar{P}'(s\cdot{I_2})=(p'_i-p_i) s^i\vect{\alpha}+\bar{P}(s\cdot{I_2})$ and
$e(g_1;g_2^{\bar{P}'(s\cdot{I_2})[j]})=e(g_1;g_2^{s^i(p'_i-p_i)\vect{\alpha}[j]})\allowbreak
\cdot{e(g_1;g_2^{\bar{P}(s\cdot{I_2})[j]})} = e(g_1;  (\bar{H}'_i[j]\cdot\bar{H}_i[j]^{-1} )^{s^i})\cdot{g_T^{\bar{P}(s)[j]}} =
e(g_1;  (\bar{H}'_i [j]\cdot\bar{H}_i [j]^{-1} )^{s^i})\cdot\bar{\mathcal{K}}[j]$
for $j=1..2$.
Finally, We use the left hand side of~\cref{prop:bivariate} and
\Cref{eq:KZG}.
Applied to $\bar{P}$, this is:
\(\bar \xi=\prod_{i=1}^d \prod_{k=0}^{i-1}
e(S_{i-k-1};\bar{H}_i)^{x_k}=\prod_{i=1}^d \prod_{k=0}^{i-1}
e(g_1^{s^{i-k-1}};{g_2}^{\bar P_i})^{r^k}\) so that
\(\bar \xi=\gen^{Q_{\bar P}(s\cdot{I_2},r\cdot{I_2})}\).
Denote by $G(Z)=\frac{Z^{d+1}-1}{Z-1}$.
Now $\bar{P}(X)={P(X)}\vect{\alpha}+{G(X\matr{\Phi})}\vect{\beta}$, then $\vect{c}={G(r\Phi)}\vect{\beta}={G(r\cdot{I_2}\Phi)}\vect{\beta}$ and
thus $\bar{P}(r\cdot{I_2})={D(\zeta)}\vect{\alpha}+\vect{c}={P(r)}\vect{\alpha}+\vect{c}$.
Therefore the verification in \Eval/\Verif~is indeed that
\(g_T^{Q_{\bar{P}}(s\cdot{I_2},r\cdot{I_2})(s-r)+\bar{P}(r\cdot{I_2})}
\checks{=}g_T^{\bar{P}(s\cdot{I_2})}=g_T^{\bar{P}(s)}\).

{\bf Complexity bounds.}
In terms of storage, apart from the public/private key pair and the
groups, the Client just has to store nine elements mod
$p$, that is $s$, $\vect{\alpha}\neq[0,0]$, $\vect{\beta}$, and
$\matr{\Phi}$, together with two group elements, $\bar{\mathcal{K}}$;
the Server has to store the polynomial ciphered thrice, the ciphered
powers of $s$ and the Merkle tree for the ciphered polynomial:
all this is $O(d)$.
In terms of communications, during the \Update~phase the Client sends
one index and three group elements, while receiving one group element
and the list of its $\log(d)$ uncles. During the \Eval/\Verif~phase, only
four elements are exchanged.
Finally, in terms of computations,
the Server performs $O(d)$ operations for the Merkle tree generation
at \Setup; fetches $O(\log(d))$ uncles at \Update; and $O(d)$
(homomorphic) operations at \Verif, thanks to~\cref{alg:prefixdiff}.
For the Client,
\Update~requires $O(\log(d))$ arithmetic operations to check the uncles
and to compute the
exponentiation $s^i$ and $\matr{\Phi}^i$, together with a
constant number of other arithmetic operations, independent of the
degree. Similarly, computing $(r\matr{\Phi})^{d_p+1}$ also requires
$O(min\{\log(d),\log(p)\})$ classical arithmetic operations thanks
to~\cref{alg:matgeomsum}. This is $\bigO{1}$ if $p$ is considered constant and the
rest is also a constant number of operations that are independent of
the degree.

{\bf Soundness.}
Let
$\left\langle{}g_1,g_1^s,g_1^{s^2},\ldots,g_1^{s^t}\right\rangle\in\GG_1^{t+1}$
be a t-BSDH instance.
For the setup phase, randomly select $\vect{\alpha},\vect{\beta},\matr{\Phi}$
and $[p_0,\ldots,p_t]$.
Then compute $W=E(P)$, $\bar{H}=g_2^{\bar{P}}$, and let
$S=\left\langle{}g_1,g_1^s,g_1^{s^2},\ldots,g_1^{s^t}\right\rangle$.
Finally homomorphically compute:
\[\bar{\mathcal{K}}=e\left(\left\langle{}g_1,g_1^s,g_1^{s^2},\ldots,g_1^{s^t}\right\rangle\odot[\bar{p}_0,\ldots,\bar{p}_t];g_2\right).\]
These inputs are indistinguishable from random inputs to the protocol
of~\cref{proto:full}.
For any number of update phases, randomly select $i$ and $p'_i$ and compute
$w'_i \gets E(p'_i)$, $\bar{H}'_i\gets g_2^{p'_i \vect{\alpha}+\matr{\Phi}^i\vect{\beta}}$ and
$\Delta=g_2^{\delta\vect{\alpha}}$.
Also compute $\mathcal{K}'={e(S_i^{(p'_i-p_i)\vect{\alpha}};g_2)}\cdot\mathcal{K}$.
Finally, select a random evaluation point $r$, compute $(\zeta,\bar{\xi})$ and
call an attacker of the \Eval{} part of the protocol to get
$(\zeta',\bar{\xi}')$ such that $(D(\zeta'),\bar{\xi}')\neq(D(\zeta),\bar{\xi})$, even
though both are passing the verification.
This means, again,
that if, on the one hand,
$D(\zeta')=D(\zeta)$, then $\bar{\xi}^{(s-r)}=\bar{\xi}'^{(s-r)}$ with
$\bar{\xi}\neq\bar{\xi}'$. Therefore $s=r$ and the secret is exposed.
If, on the other hand, $D(\zeta')\neq{D(\zeta)}$
then, as $\vect{\alpha}\neq[0,0]$, set $j\in\{1,2\}$ such that
$\alpha[j]\neq{0}$ and we have again:
\(
\left(\frac{\bar{\xi}[j]}{\bar{\xi}'[j]}\right)^{\frac{1}{\vect{\alpha[j]}(D(\zeta')-D(\zeta))}}=\gen^{\frac{1}{s-r}}
\).
This proves that the adversary would solve the t-BSDH
$\left\langle{}{-}r,\gen^{\frac{1}{s-r}}\right\rangle$ challenge.

{\bf Privacy.} We show that the protocol is hiding both $p_i$ and $\bar{p}_i$.

For $\bar{p}_i$ first.
Let $B=g_2^b$ be a DLOG instance.
For the setup phase, randomly select $s,\vect{\alpha},\matr{\Phi},d$,
$[p_0,\ldots,p_d]$ and two non-zero elements
$b_1,b_2\in\Z_p^*$.
Then compute $W=E(P)$, $\bar{H}_i=g_2^{\vect{\alpha}{p_i}}B^{\matr{\Phi}^i\Transpose{[b_1,b_2]}}$,
$S=\left\langle{}g_1,g_1^s,g_1^{s^2},\ldots,g_1^{s^t}\right\rangle$, and
$\bar{\mathcal{K}}=e(g_1;g_2^{\vect{\alpha}{P(s)}}B^{G(s\matr{\Phi})\Transpose{[b_1,b_2]}})$.
These inputs are indistinguishable from random inputs to the protocol
of~\cref{proto:full}.
For any update phase, randomly select $i$ and $p'_i$ and compute
$w'_i \gets E(p'_i)$, $\bar{H}'_i\gets g_2^{p'_i \vect{\alpha}+\matr{\Phi}^i\vect{\beta}}$
and $\Delta=g_2^{\vect{\alpha}\delta}$.
Also compute
$\bar{\mathcal{K}}'[j]=e(g_1;\Delta^{s^i})\cdot\bar{\mathcal{K}}[j]$
for $j=1..2$.
Such updates are indistinguishable from random updates to the protocol
of~\cref{proto:full}.
Randomly select any number of evaluation points $r$ and run the
associated \Eval{} phases, randomly
alternated with \Update{} phases.
Now, if an attacker can find from this transcript one
coefficient $\bar{p}_i[j]$ for $j\in\{1,2\}$, then compute
$b=(\bar{p}_i[j]-p_i\vect{\alpha}[j])/(\Phi^i\Transpose{[b_1,b_2]})[j]$
and the DLOG is revealed.

For $p_i$, we proceed with a sequence of two indistinguishable games.
  Under DLM security, cf.~\cref{def:DLM}, the
  parameter $\bar{H_i}$, or more precisely, the pair
  $(E(p_i),g_2^{{p_i}\vect{\alpha}+\Phi^i\vect{\beta}})$, is indistinguishable from
  $(E(p_i),g_2^{{p_i}\vect{\alpha}+\Gamma_i})$ for some random
  $2$-dimensional vectors~$\Gamma_i$.
  Therefore the protocol of~\cref{proto:full} is indistinguishable,
  as a whole, from the same protocol where $\matr{\Phi}^i\vect{\beta}$ is
  everywhere replaced by $\Gamma_i$, and $\vect{c}$ is (now inefficiently) computed
  as $\sum  r^i\Gamma_i$.
Now we prove that the latter is hiding.
Let $Z=E(\omega)$ be the cipher of a secret~$\omega$.
Randomly select $d$ and $[u_0,\ldots,u_d]\random\Z_p^{d+1}$.
Compute $W_i=Z\cdot{}E(u_i)=E(\omega+u_i)$.
Randomly select $\vect{\alpha}$ and $h_i$ (so that
$\Gamma_i=\log_{g_2}(h_i)-(\omega+u_i)\vect{\alpha}\in\Z_p^2$ exists, but remains unknown) for
$i=1..d$.
Randomly select $s$ and compute $\bar{\mathcal{K}}=e(g_1;H\odot[1,s,\ldots,s^d])$.
For any number of updates, randomly select $p'_i$, compute
$w'_i \gets E(p'_i)$
so that
$\delta=p'_i-p_i=(\omega+u'_i)-(\omega+u_i)=u'_i-u_i$. Thus update
$u'_i\gets\delta+u_i$ and, therefore, compute $\Delta=g_2^{\delta\vect{\alpha}}$ and
$\bar{\mathcal{K}}'[j]=e(g_1;\Delta[j]^{s^i})\cdot\bar{\mathcal{K}}[j]$
for $j=1..2$.
Alternatively run such updates with random \Eval{} phases;
all this is indistinguishable from a normal transcript of the protocol.
Now if from this transcript an attacker could find one $p_j$, then
compute $\omega=p_j-u_j$ and the encrypted value would
be revealed.\myqed
\end{proof}

\DPorTHM*
\begin{proof}

For the sake of simplicity, we here only consider the case $t=1$, that is
a single control vector.

{\bf Correctness}.
Assume that all the parties are honest.
After each update phase, thanks to the correctness of the Merkle hash
tree algorithms, we have $\Transpose{\ww}=E(\Transpose{\uu}M)$ and $\bar{\mathcal{K}}=e(g_1;g_2^{\bar{\vv} \sigma})$.
To see this, suppose a modification of the database at indices $i$ and
$k$, and let $\MM'=\MM+(\MM'_{ik}-\MM_{ik})\mathcal{E}_{ik}$ where $\mathcal{E}_{ik}$ is
the single entry matrix with $1$ at position $(i,k)$.
We have
$\Transpose{\uu}\MM'=\Transpose{\uu}\MM+\Transpose{\uu}(\MM'_{ik}-\MM_{ik})\mathcal{E}_{ik}=\Transpose{\uu}M+\gamma^i{e_k}(\MM'_{ik}-\MM_{ik})
$ where ${e_k}$ is the $k$-th canonical vector.
Thus, $\vv'= \vv+\gamma^i(\MM'_{ik}-\MM_{ik}) e_k=\vv+\delta e_k$
satisfies $\Transpose{\uu}\MM'=\Transpose{\vv'}$. Only the $k$-th coefficients are different
in $\vv$ and $\vv'$, and in $\ww$ and $\ww'$ as well. For the latter,
$\ww'_k=E(\vv'_k)=E(\vv_k+\delta)=E(\vv_k)E(\delta )=\ww_k E(\delta)$. The Server
thus computes $\ww'$ such that $\ww'=E(\Transpose{\uu}\MM')$.
Moreover, for $j=1..2$, $\bar{\vv}'[j]=\bar{\vv}[j]+  \delta\alpha[j]e_k$, so that,
similarly,
$\bar{H}'_k[j]=\Delta[j]\bar{H}_k[j]$ with $\Delta=g_2^{\delta\alpha}$,
and
$\bar{\mathcal{K}'}[j]=e(g_1;g_2^{\bar{\vv}'[j]\sigma})=e(g_1;g_2^{\bar{\vv}[j]\sigma}g_2^{\delta\alpha[j]e_k\sigma})=\bar{\mathcal{K}}[j]\cdot{e(g_1;g_2^{\delta\alpha[j]s^k})}=\bar{\mathcal{K}}[j]\cdot{e(g_1;\Delta[j]^{s^k})}$.
Now, concerning the audit phase.
Since we consider the polynomial evaluation as a dotproduct, the application of Proposition \ref{prop:bivariate} to our notations gives:
\(
(s-r) \left(\sum_{i=1}^{n-1} \sum_{k=0}^{i-1} \bar v_i s^{i-k-1} r^k\right)+ \sum_{i=0}^{n-1} \bar v_i r^i= \sum_{i=0}^{n-1} \bar v_i s^i\).
Thus we have:
\(\bar \xi=\prod_{i=1}^{n-1} \prod_{k=0}^{i-1}
e(S_{i-k-1};\bar{H}_i)^{x_k}\)
so that also  \(\bar \xi=\prod_{i=1}^{n-1} \prod_{k=0}^{i-1}
e(g_1^{s^{i-k-1}};{g_2}^{\bar v_i})^{r^k}  =\gen^{   \sum_{i=1}^{n-1}
  \sum_{k=0}^{i-1} \bar v_i s^{i-k-1} r^k  }\).

Moreover, $\alpha D(\zeta) +c= \alpha \vv x+ ((r\matr{\Phi})^{d+1}-I_2)(r\matr{\Phi}-I_2)^{-1}\vect{\beta}= \alpha \vv x+\sum_{k=0}^{n-1}r^k\Phi^k \beta=\bar{\vv} x$.
Thus we have that
\(\bar{\mathcal{K}}[j]
=
g_T^{\bar{\vv}[j]\sigma}
=
g_T^{(s-r)(\sum_{i=1}^{n-1}\sum_{k=0}^{i-1} \bar v_i[j] s^{i-k-1}
  r^k)+\bar{\vv}[j] x}
\). From the setup, this means that
\(\bar{\mathcal{K}}[j]
=
\bar{\xi}[j]^{s-r}g_T^{{D(\zeta)}\alpha[j]+c[j]}
\)
and, finally, $\Transpose{\uu}y=\Transpose{\uu}\MM x=\Transpose{\vv}x$.

{\bf Soundness}.
An attacker to the protocol must provide $(y',\zeta',\xi')$ such that
$(y',\zeta',\xi')\neq(y,\zeta,\xi)$, but still $\Transpose{\uu}y'=D_{\sk}(\bar\zeta')$,
with a non negligible advantage~$\epsilon$.
There are two cases: if
$(D_{\sk}(\bar\zeta'),\xi')\neq(D_{\sk}(\bar\zeta),\xi)$ then the attacker had to break
the polynomial evaluation; otherwise, it must be that
$\Transpose{\uu}y'=\Transpose{\uu}y$ with $y'\neq{y}$.

For the first case, \cref{thm:full} assesses the security of the
polynomial evaluation.
For the second case, we consider $T=E_{\pk}(t)$ the cipher of a secret
$t$ by the homomorphic scheme.
Here, we use again the fact that the protocol
of~\cref{protoPor} is indistinguishable as a whole from the same
protocol where, within the polynomial evaluation of, $\Phi^i\beta$
is everywhere replaced by a random~$\Gamma_i$.
Further, this is indistinguishable from a third protocol where, at
each \Write{} of index $i$, a new $\Gamma'_i$ is also randomly redrawn and
replaces $\Gamma_i$ in the Client state.
We thus continue the proof with this third game setting.
Now, using $\vect{e_\ell{}}$ the $\ell{}$-th canonical vector of $\Z_p^m$,
we can (abstractly) consider $\tilde{\uu}=\uu+t\vect{e_\ell{}}$ and
$\Transpose{\tilde{\vv}}=\Transpose{\tilde{\uu}}M=(\Transpose{\uu}+t{\vect{e_\ell{}}})M={\vv}+tM_{\ell{},*}$.
Then, for the \Init{} phase, we can randomly select $m$, $n$ and
$\ell{}\leq{m}$.
Then also $M\in\Z_p^{m\times{n}}$, $\uu\in\Z_p^m$, and compute
$\Transpose{\vv}=\Transpose{\uu}M$.
From this, compute $\ww_k=E(v_k)T^{M_{\ell{}k}}=E(v_k+tM_{\ell{}k})=E(\tilde{v}_k)$.
We also randomly select $s,\alpha$ and $\bar{H}_k$ (so that
$\Gamma_k=\log_{g_2}(\bar{H}_k)-\tilde{v}_k\alpha$ exists, but is unknown).
For any \Write{} phases, compute $\ww'_k=\ww_k T^{M'_{\ell{}k}-M_{\ell{}k}}$ and
select randomly a $\Delta$ (so that $\bar{H}'_k[j]=\bar{H}_k[j] \Delta[j]$ for $j=1..2$ now correspond
to a new $\Gamma_k'=\log_{g_2}(\bar{H}'_k)-\tilde{v}'_k\alpha$ still unknown).
Finally, the attacker provides a vector $y'$ such that both
$\Transpose{\tilde{\uu}}(y'-y)=0$ and $y'\neq{y}\mod{p}$.
Since $\ell{}$ is randomly chosen from $1..m$, the probability that the
vectors are distinct at index $\ell{}$, in other words that
$y'_\ell{}\neq{y_\ell{}}\mod{p}$, is at least $1/m$.
If this is the case, then, denoting $z=y'-y$, we have that
$z_\ell{}\neq{0}\mod{p}$.
Now, $\Transpose{\tilde{\uu}}z=0$ implies that
$\Transpose{\uu}z+tz_\ell{}=0$ so that the secret can be computed as
$t\equiv{-z_\ell{}^{-1}\cdot(\Transpose{\uu}z)}\mod{p}$ and the homomorphic
cryptosystem is subject to an attack with advantage $\epsilon/m$.
\myqed\end{proof}

\section{Paillier's cryptosystem as the linearly homomorphic primitive}\label{app:paillier}
Paillier's homomorphic system works modulo some RSA composite
number~$N$.
Now it is possible to use it to compute evaluations modulo a different
$m$ (for instance a prime), provided that $m$ is small enough:
consider the modulo $m$ operations to be over $\Z$, perform the
homomorphic operations, and use $m$ only to reduce \emph{after}
decryption. This is illustrated in~\cref{alg:paillierdp}.

\begin{algorithm}[!ht]
  \caption{Homomorphic modular polynomial evaluation with a different
    Paillier modulus}\label{alg:paillierdp}
  \begin{algorithmic}[1]
    \REQUIRE An integer $r\in[0..m-1]$;
    \REQUIRE A Paillier cryptosystem $(E,D)$ with modulus $N>(m-1)^2$.
    \REQUIRE $(E(p_0),\ldots,E(p_d))\in\Z_N^{d+1}$, such that
    $\forall{i},p_i\in[0..m-1]$ and $d<\frac{N}{(m-1)^2}-1$.
    \ENSURE $c\in\Z_N$ such that
    $D(c)\mod{m}\equiv{P(r)}\mod{m}\equiv\sum_{i=0}^d p_ir^i\mod{m}$.
    \STATE let $x_0=1$ and $c_0=E(p_0)$;
    \FOR{$i=1$ \TO $d$}
    \STATE $x_i\leftarrow{x_{i-1}\cdot{r}\mod{m}}$;\hfill\COMMENT{Now $x_i\in [0..m-1]$}
    \STATE $c_i\leftarrow{c_{i-1}}\cdot{E(p_i)^{x_i}}$;\hfill\COMMENT{Now $c=E(\sum_{k=0}^{i}p_ix_i)$}
    \ENDFOR
    \RETURN $c=c_d$.
  \end{algorithmic}
\end{algorithm}
\begin{lemma}
  \cref{alg:paillierdp} is correct.
\end{lemma}
\begin{proof}
If $0\leq{p_i}\leq(m-1)$, then as $x_i\equiv{r^i}\mod{m}$ is considered as an integer
between $0$ and $m-1$, then $0\leq\sum_{i=0}^{d}p_i
x_i\leq{(d+1)(m-1)^2}<N$ by the constraints on $d$ and $N$.
Therefore $\sum_{i=0}^{d}p_ix_i\mod{N}=\sum_{i=0}^{d}p_ix_i\in\Z$ and
now $D(c)\mod{m}=\sum_{i=0}^{d}p_ix_i\mod{m}\equiv{P(r)}$.
\myqed\end{proof}

\section{Parallel prefix-like algorithm for the
  Server}\label{app:parprefix}
We here provide the parallelization we used for the Server audits in
our experiments.
For the DPoR, the matrix-vector product part was already parallelized
in~\cite[Table~6]{Anthoine:hal-02875379}, a Server auditing the $1$TB
database in a few minutes.
For the polynomial part, as the dimensions become more rectangular, as
we can see in \cref{table:results}, the Server's polynomial part is
sometimes not negligible anymore, thus also benefits from some
parallelization.
For this, we would need to parallelize both the homomorphic
dot-product and the Horner-like pairings.
On the one hand, the former operations, line~\ref{lin:paillier:veval}
in~\cref{alg:veval}, can be blocked in independent exponentiations
and final multiplications in a binary tree.
On the other hand, for the latter operations, a standard ``baby steps
/ giant steps'' approach can be employed for the iteration of
lines~\ref{lin:begfor:veval}-\ref{lin:endfor:veval}
in~\cref{alg:veval}:
\begin{itemize}[leftmargin=2\labelsep]
\item First, for steps of size $k$, compute $t^{r^k}$, then
  $t^{r^{kj}}$ for $j=1..(d/k)$ as a parallel prefix; then iterates the
  multiplications by the coefficients of $S$ in parallel for the $d/k$
  blocks.
\item Second, then all the pairings could be computed in parallel and
  their final multiplications performed again with a binary tree.
\end{itemize}

This is exposed in~\cref{alg:parserver}.

\begin{algorithm}[htbp]\caption{Parallel Server \Eval}\label{alg:parserver}
\begin{algorithmic}[1]
\REQUIRE Group order $p$, polynomial degree $d$, evaluation point $r$
and vectors $W$, $S$, $\bar{H}[j]$, all as in~\cref{alg:veval}.
\REQUIRE Cutting parameter $q$ (e.g. can be the number of threads).
\ENSURE SERVER $\zeta$, $\bar{\xi}[j]$ for $j=1..2$.
\STATE Let $(b,\eta)\in\N^2$ s.t. $d+1=bq+\eta$, with $0\leq{\eta}<q$;
\STATE Set $b_k\gets\begin{cases}
  k(b+1)&k=0..(\eta-1)\\k{b}+\eta&k=\eta..q\end{cases}$\hfill\COMMENT{$q$ blocks of size $b+1$ or $b$}
\STATEx\COMMENT{\underline{PHASE A: $r^i\mod{p}$, for $i=0..d$}}
\STATE $\rho_0\gets{1}$, $\rho_1\gets{r}$,$i\gets{1}$;
\WHILE{$i\leq{d}$}\hfill\COMMENT{$\lceil\log_2(d)\rceil$ parallel steps}
\ParFor{$k=1..min(i;d-i)$}
\STATE $\rho_{i+k}\gets\rho_i\cdot{\rho_k}\mod{p}$;
\EndParFor
\STATE $i\gets{2i}$;
\ENDWHILE
\STATEx\COMMENT{\underline{PHASE B: $\zeta=\Transpose{W}\boxdot{x}=\prod_{i=0}^{d}w_i^{(r^i\mod{p})}$}}
\ParFor{$k=1..q$}\hfill\COMMENT{$q$ blocks of size $b$ or $b+1$
  in parallel}
\STATE\label{lin:zetak} $\zeta_k\gets\prod_{i=b_{k-1}}^{b_k-1}w_i^{\rho_i}$
\EndParFor
\STATE $\zeta\gets\prod_{k=1}^{q}\zeta_k$
\hfill\COMMENT{$\lceil\log_2(q)\rceil$ parallel steps}
\STATEx\COMMENT{\underline{PHASE C:
    $u_\ell=\prod_{k=0}^{\ell}S_{\ell-k}^{r^k}$, for $\ell=0..(d-1)$}}
\STATE $u_0\gets{S_0}$;
\FOR{$k=1$ \TO $q-1$}
\hfill\COMMENT{$q$ parallel steps}
\STATE\label{lin:ubk} $u_{b_k}\gets
u_{b_{k-1}}^{\rho_{b_k-b_{k-1}}}\prod_{\ell=b_{k-1}+1}^{b_k}
S_\ell^{\rho_{b_k-\ell}}$;
\ENDFOR
\ParFor{$k=0..(q-1)$}\hfill\COMMENT{$q$ blocks of size $b$ or $b-1$ in parallel}
\FOR{$\ell=0$ \TO $b_{k+1}-b_{k}-1$}
\STATE\label{lin:ubl} $u_{b_k+\ell+1}\gets{S_{b_k+\ell+1}}\cdot{u_{b_k+\ell}^r}$;
\ENDFOR
\EndParFor
\STATEx\COMMENT{\underline{PHASE D: $\bar{\xi}=\prod_{i=1}^d \prod_{k=0}^{i-1}e(S_{i-1-k};\bar{H}_i)^{r^k}$}}
\STATE $\bar{\xi}=\Transpose{[1_{\GG_T},1_{\GG_T}]}\in\GG_T^2$;
\FOR{$j=1$ \TO $2$}
\ParFor{$k=1..q$}
\hfill\COMMENT{$q$ blocks of size $b$ or $b+1$ in parallel}
\STATE\label{lin:ehu} $\bar{\xi}_k[j]\gets\prod_{\ell=b_{k-1}}^{b_k-1}e(u_\ell;\bar{H}_{\ell-1}[j])$
\EndParFor
\STATE $\bar{\xi}[j]\gets\prod_{k=1}^{q}\bar{\xi}_k[j]$
\hfill\COMMENT{$\lceil\log_2(q)\rceil$ parallel steps}
\ENDFOR
\end{algorithmic}
\end{algorithm}

\begin{lemma}
\cref{alg:parserver} is correct, work-optimal with work $W_q=O(d)$ and runs
in time $W_q/q+o(W_q)$ on $q$ processors.
\end{lemma}
\begin{proof}
Correctness of phases A, B and D stems directly from the correctness
of~\cref{alg:veval}. Phase C is correct since the new variables $u_\ell$
satisfy $\{u_0=S_0, u_{\ell+1}=S_{\ell+1}u_\ell^r\}$.

Then, $p$ is the prime group order, and for any homomorphic system
satisfying~\cref{eq:homo:addmul} we have:
\begin{itemize}[leftmargin=2\labelsep]
\item Phase A: requires $d$ multiplications modulo $p$ with depth
  $O(\log(d))$ and the parallelism is thus only bounded by Brent's law~\cite[Lemma~2]{Brent:1974:law};
\item Phase B: requires $d+1$ homomorphic exponentiations and $d$ homomorphic multiplications with a depth of $b=d/q$ such operations and the
  parallelism is thus only bounded by Brent's law;
\item Phase C: requires $d$ exponentiations and multiplications
  in $\GG_1$. But this is implemented in parallel with a depth of
  $b=d/q$ such operations, only after precomputing $q-1$ times $b$
  operations each with a depth of $\log(b)$;
\item Phase D: requires $d$ pairings and $d-1$ multiplications in
  $\GG_T$ with a depth of $b=d/q$ such operations and the
  parallelism is thus only bounded by Brent's law.
\end{itemize}

So only Phase C requires more operations in parallel than in sequence.
And that number of operations is $d+b(q-1)$ exponentiations and
multiplications if ran on $q$ processors.
This latter work is in fact optimal for prefix-like computations as
shown in \cite[Corollary~4]{Snir:1986:prefix} (see
also~\cite{Roch:2006:prefix}): indeed consider a family of binary
gates $\theta_{\rho_i}(a,b)$ that on inputs $a$ and $b$ compute
$a\cdot{b^{\rho_i}}$, that is one multiplication and one exponentiation. They
satisfy the conditions of \cite[Corollary~4]{Snir:1986:prefix} and
thus computing all the $u_\ell$ is lower bounded by $d(2-1/q)$ calls to
that gate when ran on $q$ processors.
\end{proof}

\begin{remark} The accumulated \emph{independent}
  exponentiations/pairings of lines \ref{lin:zetak}, \ref{lin:ubk} and
  \ref{lin:ehu} of~\cref{alg:parserver} can in fact be gathered in
  small batches, where each batch can factorize some computations
  (e.g. using a generalized Shamir trick with multiple exponentiations
  in $\GG_1$, or using NAF windows, etc.). Therefore, on the one hand,
  with respect to a purely sequential computation, the extra work
  required by Phase C (when used with more than $2$ processors) is in
  fact batched.
  On the other hand, the other part of Phase C cannot benefit from
  these batches and is therefore dominant, but is more parallel.
  Therefore, as shown also in~\cref{tab:parallel}, this allows us to
  reach, on multiple cores, pretty good overall practical speed-ups.
\end{remark}

\begin{table}[!ht]\centering\caption{Parallel Server-side VESPo}\label{tab:parallel}
\begin{tabular}{crrrrrr}
\toprule
Degree & 5816 & 18390 &  58154  & 186093 & 426519 & 4026778 \\
\midrule
1 core& $4.4$s & $13.5$s & $42.6$s & $141.7$s & $324.1$s & $3\,064.8$s  \\
4 cores & $1.2$s & $3.8$s & $11.8$s &   $38.3$s &  $87.8$s &  $831.6$s \\
8 cores & $0.7$s & $2.0$s &   $6.3$s &  $19.9$s &  $45.4$s &  $428.9$s \\
12 cores &$0.5$s &  $1.4$s &  $4.2$s &  $13.4$s &  $30.6$s &  $283.6$s \\
\bottomrule
\end{tabular}
\end{table}

This parallelism can be used to further reduce the Server latency
for large databases, to allow faster multi-user queries, and thus to make
the scheme even more practically relevant.

\section{Post-quantum homomorphic routines}\label{app:pqc}
The use of linearly homomorphic encryption (e.g., Paillier) and pairings
means that, as implemented, our protocols are not resistant to quantum
attacks. In response to recent recommendations by NIST and other
standards organizations that all cryptographic solutions be made
quantum-resistant, we considered the impacts of replacing these
primitives with fully homomorphic encryption (FHE) primitives which are
believed to be quantum-resistant.

Unfortunately, there are two reasons why further work is needed before
we could recommend using FHE in our protocols. First, as we discuss in
detail below, our preliminary implementation results are prohibitively
slow using state of the art FHE libraries, due apparently to the
inherent non-linear nature of polynomial evaluation on encrypted
evaluation points. Second, our proof of security as written reduces the
soundness guarantee to the $t$-BSDH problem, which has no analogue in
FHE cryptosystems, and it is not clear what different assumption on FHE
primitives could be used in its place.

We now detail our preliminary investigation into using FHE in our
protocols, to better explain the shortcomings mentioned above and
hopefully encourage future work in this direction.

We need two systems.
First, where we use Paillier's cryptosystem, our protocols were
already abstracted by the requirements of~\cref{eq:homo:addmul}.
It is thus possible to use instead any quantum-safe linearly
homomorphic primitives.
There, Paillier's routine with larger parameters might be a possibility, see
e.g.~\cite{Bernstein:20217:PQRSA}. Other possibilities for now is to
use quantum-safe fully homomorphic encryption, like
BGV~\cite{Brakerski:2014:BGV}, here without bootstrapping.

\newcommand{\pE}{\ensuremath{\mathcal{E}}}
\newcommand{\pD}{\ensuremath{\mathcal{D}}}

Second, we need to modify the parts where we use pairings, in order to
replace them with quantum-safe routines.
For this we first abstract the requirements.
Denote by \pE, and resp. \pD, the homomorphic
encryption, resp. decryption, functions.
We want those to support
homomorphic addition,
homomorphic multiplication between a ciphered message and a cleartext,
together with depth-1 homomorphic multiplication between two
ciphertexts and with homomorphic equality testing (in a private
setting, this latter requirements can also for instance be implemented
by decryption and direct equality testing).
We can notate these requirements as follows:
\begin{eqnarray*}
  \pD(\pE(m_1)\oplus\pE(m_2))	&=& m_1 + m_2 \\
  \pD(\pE(m_1)^{m_2})	&=&  m_1 \times m_2	\\
  \pD(\pE(m_1)\otimes\pE(m_2))	&=& m_1 \times m_2\\
  \pE(m_1)\checks{\oeq}\pE(m_2)	&\iff&	  m_1\checks{=}m_2
\end{eqnarray*}

The pairings parts
in~\cref{proto:full} are now transformed as in~\cref{tab:abspq}:
\begin{table}[htb]\caption{Abstraction of the pairings
    functionalities}\label{tab:abspq}
\begin{tabular}{cc}
\toprule
{\Setup} & $\bar{\mathcal{K}}\gets\pE(\bar{P}(s))$,
$S\gets[\pE(s^k)]_{k=0}^{d-1}$,\\
 & $\bar{H}\gets[\pE({p_i}\vect{\alpha}+\matr{\Phi}^i\vect{\beta}]_{i=1}^{d}$\\
\midrule
{\Update} & $\bar{\mathcal{K}}\gets \pE\left(\alpha(p'_i-p_i)s^i\right)\oplus\bar{\mathcal{K}}$\\
{\Eval} & $\bar{\xi}\gets\oplus_{i=1}^d\left(\oplus_{k=0}^{i-1}S_{i-k-1}^{x_k}\right)\otimes\bar{H}_i$\\
{\Verif} & $\bar{\xi}^{s-r}\oplus\pE(D(\zeta)\vect{\alpha}+\vect{c})\checks{\oeq}\bar{\mathcal{K}}$\\
\bottomrule
\end{tabular}
\end{table}

The two important issues are then the security analysis and the
performance.
First, the security analysis we have performed depends on assumptions
of pairings (namely the hardness of $t$-BSDH).
For \cref{tab:abspq} we would need to use some other assumption on
the chosen FHE primitives.
Second, our protocol efficiency crucially depends on efficient
ciphertext-cleartext multiplication.
We here report on some attempts with the BGV system implemented with
the SEAL~\cite{Microsoft:2022:seal4.0} and the
HElib~\cite{Halevi:2021:HElib} libraries.

We were able to make our protocol work with these LWE-like
implementations but for now there is a prohibitive performance price to
pay, for two reasons:
\begin{enumerate}
\item A first constraint in SEAL and HElib is the size of the
  cleartext modulus which can usually not yet be very large, in
  practice at most some small fraction of a machine word.
\item A second limitation for these libraries, is that the
  ciphertext-cleartext multiplication is not much more efficient
  than ciphertext-ciphertext, since the noise in the polynomials is
  similarly increasing in both cases.
\end{enumerate}

More precisely, for the computation of our coefficient $\zeta$, we were
able to use batched arithmetic with both SEAL and HElib and this is
quite efficient, but works only for very small primes.

Differently, this is for the computation of our second
coefficient,~$\xi$, that the price to pay is much too prohibitive,
even for very
small primes and (too) low security parameters.
Indeed, to compute $\xi$, our Server scheme involves computations of
the form $S_3\oplus{S_2^r}\oplus{S_1^{r^2}}\oplus{S_0^{r^3}}$, where
$r$ is a cleartext and the $S_i$ are ciphertexts.
On the one hand, if the $r^k$ are precomputed, this is of
constant multiplicative depth $1$, even when counting
ciphertext-cleartext multiplications, but then the overall double-loop
scheme of~\cref{prop:bivariate} is quadratic-time.
On the other hand, if $\xi$ is instead homomorphically computed with the
linear prefix-like~\cref{alg:prefixdiff}, the computations now involve
in fact computations of the form
$S_3\oplus(S_2\oplus(S_1\oplus(S_0^r))^r)^r$. As mentioned, even
though the involved multiplications are only ciphertext-cleartext,
in the available libraries the noises increase linearly, much
closer to a linear multiplicative depth. Bootstrapping is thus
required a linear number of times.
For instance, the latency of a BGV bootstrapping operation
costs at least several dozen
seconds~\cite{Halevi:2021:HElib}\footnote{In contrast, some other
  libraries, such as FHEW~\cite{Ducas:2015:FHEW} and
  TFHE~\cite{Chillotti:2021:TFHE}, may have faster bootstrapping
  operations but require to re-implement the homomorphic arithmetic
  with boolean circuits.}.
We provide in~\cref{tab:pqcbench}, evaluations of our
scheme using either SEAL and the quadratic, depth-$1$ version, or
HElib and the linear, but bootstrapped version.
Comparing with~\cref{tab:lintests}, we see that quantum-safe routines
are for now still several orders of magnitude slower.

\newcommand{\lrcmidruleTF}{\cmidrule(lr){3-5}} %
\begin{table}[htbp]\caption{Post-quantum prototypes
{\footnotesize(SEAL modulo $1032193$, with $123.1$ eq. security, $4096$-batched $\zeta$,
and quadratic-time $\xi$;
HElib modulo $31$, with $39.5$ eq. security, $24$-batched $\zeta$,
and linear-time but bootstrapped $\xi$)}.}\label{tab:pqcbench}
\begin{tabular}{ccrrrr}
\toprule
& \multirow{2}{*}{Deg.} & \multicolumn{3}{c}{Server} & \multirow{2}{*}{Client}
\\
\lrcmidruleTF
& & $\zeta$ & $\xi$ & bootstrap. &
\\
\midrule
\multirow{8}{*}{\rotatebox[origin=c]{90}{SEAL}}
         & 32 & $<$0.01s & 2.03s & 0 &  \multirow{8}{*}{6.57ms} \\
         & 64 & $<$0.01s & 5.45s & 0 &  \\
         & 128 & $<$0.01s & 20.95s & 0 &  \\
         & 256 & $<$0.01s & 82.14s & 0 &  \\
         & 512 & $<$0.01s & 325.87s & 0 &  \\
         & 1024 & $<$0.01s & 1\,294.88s & 0 &  \\
         & 2048 & $<$0.01s & 5\,171.84s & 0 & \\
         & 4096 & $<$0.01s & 20\,667.99s & 0 &  \\
\midrule
\multirow{8}{*}{\rotatebox[origin=c]{90}{HElib}}
         & 32 & 0.01s & 7.26s & 0 & \multirow{8}{*}{283.44ms} \\
         & 64 & 0.01s & 80.02s & 13 &  \\
         & 128 & 0.02s & 257.20s & 45 &  \\
         & 256 & 0.03s & 613.83s & 109 &  \\
         & 512 & 0.05s & 1\,334.52s & 238 &  \\
         & 1024 & 0.10s & 2\,765.61s & 493 &  \\
         & 2048 & 0.20s & 5\,643.51s & 1005 &  \\
         & 4096 & 0.39s & 11\,382.50s & 2030 & \\
\bottomrule
\end{tabular}
\end{table}

The dominant cost in these experiments is in fact the bootstrapping.
Future work thus might be:
\begin{itemize}
\item Designing a post-quantum linearly homomorphic encryption
  with efficient ciphertext-cleartext multiplication
\item Transforming the computation of $\xi$ so that it is more
  batchable (a strategy could be to start by adapting the
  parallelization presented in~\cref{app:parprefix}, so that more
  identical operations could be performed simultaneously)
\end{itemize}

For instance, phase C in~\cref{alg:parserver} is solely responsible
for the multiplicative depth. Then we see that line~\ref{lin:ubk}
can be performed with a depth of $q$, while line~\ref{lin:ubl} can be
performed with $q$ depth\nobreakdash-$b$ operations, with $d=bq$.
With some FHE implementations (as reflected in~\cref{tab:pqcbench})
the first aggregated multiplications require less bootstrapping.
Thus, depending on their respective costs and the actual
architecture, some choices of $b$ (and $q$) might reduce the overall
required bootstraps.
By looking at~\cref{tab:pqcbench,tab:lintests}, we see that even such
a (small) constant gain in bootstrapping is not yet sufficient to compete
with the pairings.

\end{document}